\documentclass[sigplan,9pt]{acmart}

\acmConference[LICS'18]{Thirty-Third Annual ACM/IEEE Symposium on Logic in Computer Science (LICS)}{July 09--12, 2018}{Oxford}
\acmYear{2018}
\acmISBN{} 
\acmDOI{} 
\startPage{1}

\setcopyright{none}

\usepackage{subcaption}
\usepackage{tikz}
\usepackage{hyperref}
\usepackage{amsthm, amsfonts}
\usepackage{amsmath}
\usepackage{stmaryrd}
\usepackage{mathtools}
\usepackage{framed}
\usepackage{changebar} 

\usepackage{todonotes}
\usepackage{amssymb}

\usepackage[latin1]{inputenc}

\newcounter{thm}
\newcounter{theorem}

\theoremstyle{definition}
\newtheorem{defin}[thm]{Definition}
\newtheorem{prop}[thm]{Proposition}
\newtheorem{llemma}[thm]{Lemma}
\newtheorem{theorem}[thm]{Theorem}
\newtheorem{corollary}[thm]{Corollary}

\newtheorem{example}[thm]{Example}
\newtheorem{remark}[thm]{Remark}

\newcommand{\N}[0]{{\mathbb{N}}}

\newcommand{\La}[0]{{\mathcal{L}}}

\newcommand{\Ff}[0]{{\mathfrak{F}}}
\newcommand{\Gg}[0]{{\mathfrak{G}}}

\newcommand{\carrow}[3]{#3 \xleftarrow{#2} #1 }

\newcommand{\Pa}[0]{{\mathcal{P}}}

\newcommand{\Pirn}[0]{\Pi Z_n}
\newcommand{\Rn}[0]{Z_n}
\newcommand{\Rnm}[0]{Z_{n-1}}
\newcommand{\Trace}[0]{\operatorname{Traces}}

\newcommand{\height}[0]{\text{height}}

\newcommand{\HArr}[0]{{\operatorname{HArr}}}
\newcommand{\Arr}[0]{{\operatorname{Arr}}}
\newcommand{\Obj}[0]{{\operatorname{Obj}}}

\newcommand{\freedisth}[0]{{H_\Sigma^D}}

\newcommand{\mine}[1]{}

\title{Wreath Products of Distributive Forest Algebras}

\author{Michael Hahn}
\affiliation{%
  \institution{Stanford University}
}
\email{mhahn2@stanford.edu}

\author{Andreas Krebs}
\affiliation{%
  \institution{University of T{\"u}bingen}
}
\email{mail@krebs-net.de}

\author{Howard Straubing}
\affiliation{%
  \institution{Boston College}
}
\email{howard.straubing@bc.edu}

\begin{document}

\begin{abstract}

It is an open problem whether definability in Propositional Dynamic Logic (PDL) on forests is decidable.
Based on an algebraic characterization by Boja{\'n}czyk, {\it et. al.,} (2012) in terms of forest algebras, Straubing (2013) described an approach to PDL based on a $k$-fold iterated distributive law.
A proof that all languages satisfying such a $k$-fold iterated distributive law are in PDL would  settle decidability of PDL.
We solve this problem in the case $k = 2$: All languages recognized by forest algebras satisfying a 2-fold iterated distributive law are in PDL.
Furthermore, we show that this class is decidable.
This provides a novel nontrivial decidable subclass of PDL, and demonstrates the viability of the proposed approach to deciding PDL in general.
\end{abstract}

\maketitle


\section{Introduction}

\subsection{Motivation}

A much-studied problem in the theory of automata is that of determining whether a given regular language $L$ can be defined by a formula of some logic--in other words, to give an effective characterization of the precise expressive power of the logic.  For automata over words, there is by now a large collection of such results, giving effective tests for definability in many temporal and predicate logics. 

For tree automata, the situation is quite different: the problems of effectively deciding expressibility in $CTL,$ $CTL^*,$ first-order logic with ancestor, and propositional dynamic logic ($PDL$) remain open to this day.

In  \cite{bojanczyk-wreath-2012} Boja{\'n}czyk, {\it et. al.,} proposed to attack this problem by adapting the algebraic tools that have proved so successful in the case of word languages.  They proved (working in the setting of languages of finite unordered forests) that the languages definable in each of the four logics cited above can be characterized as those recognized by iterated wreath products of forest algebras, where the factors in the wreath product all belong to a particular decidable variety of algebras.  For example, languages in $PDL,$   which are the focus of the present paper, are exactly those recognized by wreath products of forest algebras, each of which has an idempotent and commutative horizontal monoid, and which satisfies a distributive law.  (See below for precise definitions).

Straubing, in ~\cite{straubing-new-2013} detailed a possible approach to $PDL$ by noting that forest algebras that divide an iterated wreath product of $k$ distributive algebras satisfy a kind of order $k$ generalized distributive law (analogous to solvable groups, which satisfy an order $k$ commutative law, for some $k>0$).  Determining whether a given forest algebra satisfies such a generalized distributive law for some $k$ is a decidable problem.  So if one could prove that every such generalized distributive forest algebra divides a wreath product of distributive algebras, the question of definability in $PDL$ would be settled.

Here, we solve this problem in the case $k=2$.  More precisely, we show that every 2-distributive finite forest algebra divides a wreath product of four distributive algebras, and that further, 2-distributivity is itself a decidable property. Thus we have identified a decidable nontrivial subclass of $PDL,$ and demonstrated the viability of the proposed approach for deciding $PDL$ in general.

$PDL$ contains the logics $CTL$ and $CTL^*$, the latter being the bisimulation-invariant part of first-order logic on trees.
The graded equivalent of $PDL$, also known as Chain Logic, fully contains first-order logic with ancestor.
$PDL$ and Chain Logic are the largest among the tree logics that have been considered in \cite{bojanczyk-wreath-2012} and related work on finding effective characterizations.
Indeed, $PDL$ could be seen as the `largest' nontrivial bisimulation-invariant  logic on finite trees: It seems that no nontrivial logic class has been found between $PDL$ and the bisimulation-invariant Boolean Formula Value problem, which is not definable in any of these logics \cite{potthoff-first-order-1995, straubing-new-2013}.
Strikingly, for all of these logics, decidability is still open despite several attempts, and a family of decidability results obtained for smaller fragments of these logics (e.g., \cite{bojanczyk-effective-2008,bojanczyk-piecewise-2012,place-deciding-2010,benedikt-regular-2009}).
Furthermore, all of these logics were characterized in \cite{bojanczyk-wreath-2012} in terms of iterated wreath products of certain forest algebras that satisfy a distributive law.
While the representations of $CTL$ and $CTL^*$ place restrictions on these algebras, $PDL$ is characterized by products of arbitrary distributive forest algebras.

\subsection{Outline of the paper}

In Section~\ref{sec:background} we recall the basic definitions of forests and forest algebras.
In Section~\ref{sec:wreath-pdl}, we review the algebraic characterization of Propositional Dynamic Logic (PDL) in terms of wreath products of distributive forest algebras.
In Section~\ref{sec:2-dis}, we discuss 2-distributive forest algebras, the main object of study in this paper.
In Section~\ref{sec:categories}, we review a generalization of the classical Derived Category Theorem to the setting of forest algebras, recently introduced by \cite{straubing-forest-2018}.
In Section~\ref{sec:main-result}, we prove the main result: Languages recognized by 2-distributive forest algebras are in PDL. We will discuss the two contributions on which this result rests, namely a separation theorem and a Local-Global theorem.
We discuss the role of our results in Section~\ref{sec:discussion}.

\section{Background on Finite Forests and Forest Algebras}\label{sec:background}

\begin{defin}[Forest Algebras]\label{def:forest-algebra}
A tuple $(H,V)$ is called a \emph{forest algebra} if the following conditions hold:
\begin{enumerate}
\item $H$ is a monoid, whose operation is written $+$, with neutral element $0_H$

\item $V$ is a monoid, whose operation is written $\cdot$, with neutral element $1_V$

\item There is an operation $V \times H \rightarrow H$, also written $\cdot$
\item This operation is an action: $v\cdot (v'\cdot h) = (vv')\cdot h$

\item The action is faithful: If $v\cdot h = v'\cdot h$ for all $h \in H$, then $v=v'$

\item There is a map $I_\cdot : H \rightarrow V$ such that $I_h h' = h+h'$.

Note that, due to faithfulness, this map is uniquely determined. We will write $h+v$ for $I_h \cdot v$ (that is, the product of $I_h$ and $v$ as elements of $V$).

\item For each $h \in H$, there is $v \in V$ such that $h = v\cdot 0_H$.
\end{enumerate}

If $(H,V)$, $(H',V')$ are forest algebras, then a pair $\phi = (\phi_H,\phi_V)$ of maps $\phi_H : H \rightarrow H'$, $\phi_V : V \rightarrow V'$ is called a \emph{morphism} if it respects these structures.
More formally, we require that (1) $\phi_H$, $\phi_V$ are monoid morphisms, (2) $\phi_V(v)\phi_H(h) = \phi_H(vh)$, (3) $I_{\phi_H(h)} = \phi_V(I_h)$ for all $h \in H$, $v \in V$.
\end{defin}


We will often omit the $\cdot$ operator for the multiplication on $V$ and the action of $V$ on $H$.
But we will never omit the $+$ operator for the addition on $H$.

We will use $\Ff, \Gg$ as variables for forest algebras.
For a forest algebra $\Ff = (H,V)$, the elements of $H$ are called \emph{forest types}, while the elements of $V$ are called \emph{context types}.
Given a forest algebra $\Ff = (H,V)$, we will sometimes write $H_\Ff$, $V_\Ff$ for $H$ and $V$, respectively.


\paragraph{Trees and Forests}
Let $\Sigma$ be a finite set, referred to as \emph{alphabet}.
By \emph{trees over $\Sigma$}, we refer to finite (well-founded) trees, all of whose nodes are labeled with symbols in $\Sigma$.
We do not allow empty trees.



\paragraph{Contexts, Free Forest Algebra}
A \emph{context} is a forest where (exactly) one leaf is labeled with a variable instead of a symbol from $\Sigma$.

Contexts form a monoid $V_\Sigma$: We define $v \cdot v'$ to be the context obtained by replacing the variable in $v$ with the context $v'$. The result is again a context.
This operation is associative.
The identity element is the context consisting of only a variable.
Forests form a monoid $H_\Sigma$, with union as the monoid operation $+$, and the empty forest as the identity element.
The monoid of contexts acts on the monoid of forests, with insertion of forests into the hole of a context as the operation.
Taken together, the monoid of forests and the monoid of contexts form a forest algebra, the \emph{free forest algebra} $\Sigma^\Delta = (H_\Sigma, V_\Sigma)$.

\begin{defin}[Recognition]
A forest algebra $(H,V)$ \emph{recognizes} a forest language $\La \subseteq H_\Sigma$ (that is, a set of forests) if and only if there is a forest algebra morphism $\phi : \Sigma^\Delta \rightarrow (H,V)$ such that $\La = \phi^{-1}(\phi(\La))$.
\end{defin}

Many notions from Universal Algebra carry over to forest algebras.
If a tuple $(H',V')$ consists of subsets $H' \subseteq H$ and $V' \subseteq V$, then it is a \emph{subalgebra} of a forest algebra $(H,V)$ if it is closed under the forest algebra operations of $(H,V)$.
A pair of equivalence relations on $H$ and $V$ is called a \emph{congruence} of $(H,V)$ if it respects the forest algebra operations.
The \emph{quotient} $(H',V')$ of $(H,V)$ by a congruence is formed by taking $H'$ and $V'$ to be the sets of equivalence classes of the respective equivalence relations given by the congruence.
Since congruences respect forest algebra operations, the action $\cdot$ and the operation $I_\cdot$ are well-defined on the quotient.

\begin{defin}[Division]
Let $(H,V)$, $(H', V')$ be forest algebras. Then $(H,V) \prec (H', V')$ if $(H,V)$ is a quotient of a subalgebra of $(H',V')$.
\end{defin}

If $\Ff \prec \Ff'$, then any language recognized by $\Ff$ is also recognized by $\Ff'$. This is shown in analogy to the parallel result for word languages and monoids~\cite{eilenberg-automata-1976}.

\paragraph{Horizontal Idempotency and Commutativity}
A forest algebra $(H,V)$ is \emph{horizontally commutative and idempotent} if $h+h = h$ and $h_1+h_2 = h_2+h_1$ hold for all $h, h_1, h_2 \in H$.
From now on, we will assume that all forest algebras are horizontally commutative and idempotent.
This is no loss of generality, since all PDL languages are recognized by horizontally commutative and idempotent forest algebras \cite{bojanczyk-wreath-2012}.

Furthermore, we will consider \emph{trees} without regard to order and multiplicity of children. 
More formally, we can define trees inductively as follows: The set of trees over $\Sigma$ is the smallest set such that (1) $\alpha[\emptyset]$ is a tree whenever $\alpha\in\Sigma$, and (2) whenever $C$ is a finite set of trees, and $\alpha \in \Sigma$, then $\alpha[C]$ is also a tree.
%
Note that this means that the free forest algebra $\Sigma^\Delta$ is also horizontally commutative and idempotent.

\section{PDL and Wreath Products of Forest Algebras}\label{sec:wreath-pdl}

Having introduced the general algebraic background for studying wreath products of forest algebras,
we now discuss \emph{distributive} forest algebras and the forest logic that we are focusing on, Propositional Dynamic Logic. 
We refer to \cite{bojanczyk-wreath-2012} for the definition of Propositional Dynamic Logic as a temporal logic on trees and forests.
For our purposes, the algebraic characterization from \cite{bojanczyk-wreath-2012} will be sufficient.
We will first introduce the forest algebra wreath product.

\subsection{Wreath Products of Forest Algebras}
In \cite{bojanczyk-wreath-2012}, Boja{\'n}czyk \textit{et al.} introduced the following wreath product operation on forest algebras: 

\begin{defin}[\cite{bojanczyk-wreath-2012}]
Let $(H_1,V_1)$, $(H_2,V_2)$ be forest algebras.
Then define the \emph{wreath product} as
$$(H_1,V_1) \wr (H_2,V_2) := (H_1 \times H_2, V_1^{H_2} \times V_2)$$
with the following operations: For $(f,v) \in V_1^{H_2} \times V_2$ and $(h_1, h_2) \in H_1 \times H_2$, let $$(f,v)  (h_1, h_2) := (f(h_2)h_1, vh_2)$$
For $(f,v), (f',v') \in V_1^{H_2} \times V_2$, let $$(f,v)(f',v') := (f'', vv')$$ with $f''(h) := (f(v'h)) \cdot (f'(h))$. For the operation on $H_1\times H_2$, we use the structure of the direct product.
\end{defin}

\cite{bojanczyk-wreath-2012} showed that $(H_1, V_1) \wr (H_2,V_2)$ is a forest algebra.
Also, the wreath product is associative up to isomorphism \cite{bojanczyk-wreath-2012}: $((H_1, V_1) \wr (H_2, V_2)) \wr (H_3, V_3) \equiv(H_1, V_1) \wr ((H_2, V_2) \wr (H_3, V_3))$.
Therefore, it makes sense to talk about iterated wreath products of classes of forest algebras.

The wreath product has been applied to finite forest algebras in previous work, but nothing in the definition precludes application to infinite forest algebras ($V_1^{H_2}$ will then be uncountable).
We will make reference to wreath products of infinite forest algebras for illustrative purposes, but our main result will not depend on this.


\subsection{Distributive Forest Algebras and PDL}\label{sec:dis}

A forest algebra $(H,V)$ is called \emph{distributive} \cite{bojanczyk-wreath-2012} if 
\begin{equation}
v(h_1+h_2) = vh_1 + vh_2
\end{equation}
 for all $v \in V$, $h_1, h_2 \in H$.
Equivalently, $(H,V)$ is distributive if, for each morphism $\phi : \Sigma^\Delta \rightarrow (H,V)$, for all contexts $c$, and for all forests $f_1, f_2$, the following equality holds:
\begin{equation}\label{eq:dis-def-2}
\phi(c(f_1+f_2)) = \phi(cf_1+cf_2)
\end{equation}
Recall that, additionally, we require idempotency ($h+h = h$ for $h \in H$) and commutativity ($h_1+h_2=h_2+h_1$ for $h_1, h_2 \in H$) for all forest algebras in this paper.

There is a close connection between distributive forest algebras and the sets of paths of forests.
If $\phi : \Sigma^\Delta \rightarrow \Ff$ is a morphism to a distributive algebra $\Ff$, then, for any forest $f$, the value $\phi(f)$ only depends on the set of (not necessarily maximal) paths in the forest $f$.
If $L$ is a regular language of words, then the language of forests that have (not necessarily maximal) paths in $L$ is recognized by a finite distributive forest algebra.
More generally, the class of forest languages recognized by such algebras is exactly the Boolean algebra generated by languages of this form (Proposition~\ref{prop:distr-char}). 

\begin{defin}
If $f$ is a forest, we write $\pi(f)$ for the set of its (not necessarily maximal) paths, starting at the root. Thus, $\pi(f)$ is a finite subset of $\Sigma^*$ closed under taking prefixes ($wv \in X \Rightarrow w \in X$).
\end{defin}


\begin{prop}\label{prop:distr-char}[Theorem 5.3 from \cite{bojanczyk-wreath-2012}]
A language $\La \subseteq H_\Sigma$ is recognized by a finite distributive algebra if and only if it is a finite Boolean combination of languages of the form $$\La_I := \{f \in H_\Sigma : \pi(f) \cap I \neq \emptyset\}$$ with $I \subseteq \Sigma^*$ regular word languages.
\end{prop}



We can now state the algebraic characterization of Propositional Dynamic Logic (PDL) by \cite{bojanczyk-wreath-2012}:

\begin{theorem}[\cite{bojanczyk-wreath-2012}]
A regular language $\La \subset H_\Sigma$ is definable in PDL if and only if there are finite distributive forest algebras $\Ff_1, ..., \Ff_k$ such that $\Ff_1 \wr ... \wr \Ff_k$ recognizes $\La$.
\end{theorem}

By this result, the problem of deciding definability of a language $\La$ in PDL can be reduced to the problem of determining whether it is recognized by an iterated wreath product of finite distributive algebras.
This, in turn, is equivalent to the question whether the syntactic forest algebra of $\La$ divides such a product.

\section{2-Distributive Forest Algebras}\label{sec:2-dis}

We now define 2-distributive forest algebras, which will be the subject of our main result.
In the Discussion (Section~\ref{sec:discussion}), we will discuss how this notion and results in this section generalize to $k > 2$, and how this notion relates to the general approach to settling decidability of PDL.
A forest algebra $(H,V)$ is \emph{2-distributive} if, for any alphabet $\Sigma$ and for all morphisms $\phi : \Sigma^\Delta \rightarrow (H,V)$, for all contexts $v \in V_\Sigma$, and for all forests $f_1, f_2 \in H_\Sigma$ with $\pi(f_1) = \pi(f_2)$, the following equality holds:
$$\phi(v(f_1+f_2)) = \phi(vf_1 + vf_2)$$
That is, we take the same condition as for distributivity, but require this only in the case when $\pi(f_1) = \pi(f_2)$.
In this sense, we are describing a 2-fold iterated distributive law.
As before, we further require horizontal idempotency ($h+h = h$ for $h \in H$) and commutativity ($h_1+h_2=h_2+h_1$ for $h_1, h_2 \in H$).

\begin{remark}
By definition, the property of 2-distributivity is expressed by a collection of identities between explicit operations, and thus 2-distributive forest algebras form a (Birkhoff) variety. We will not explicitly make use of varieties here.
However, it deserves mentioning that most classes of forest languages that have been characterized using identities involve \emph{profinite} identities involving implicit operations \cite{bojanczyk-effective-2008}, while defining 2-distributivity does not involve such implicit operations.
\end{remark}

We will now show that 2-distributive algebras are closely connected to wreath products of distributive forest algebras.
The following proposition is not hard to prove:
\begin{prop}\label{prop:wreath-2}
If $\Ff_1, \Ff_2$ are distributive forest algebras, then $\Ff_1 \wr \Ff_2$ is 2-distributive.
\end{prop}

\begin{proof}
Let $\Ff_1 = (H_1,V_1)$, $\Ff_2 = (H_2, V_2)$ be distributive forest algebras (finite or infinite).
Take any alphabet $\Sigma$ and a morphism $\phi : \Sigma^\Delta \rightarrow \Ff_1 \wr \Ff_2$.
Let $v \in V_\Sigma$, $f_1, f_2 \in H_\Sigma$ with $\pi(f_1) = \pi(f_2)$.
Note that context types and forest types in $\Ff_1 \wr \Ff_2$ are tuples, whose right elements are elements of $\Ff_2$.
As in the case of monoid wreath products, it is easy to see that the projection of elements of $\Ff_1 \wr \Ff_2$ on the second component is a morphism from $\Ff_1 \wr \Ff_2$ to $\Ff_2$.
That is, $\pi^{(2)}\circ \phi : \Sigma^\Delta \rightarrow \Ff_2$.
Since $\Ff_2$ is distributive and $\pi(f_1) = \pi(f_2)$, we know $\pi^{(2)}\phi(f_1) = \pi^{(2)}\phi(f_2)$. 
Let's call this element $h_0$.
Let us compute $\phi(v(f_1+f_2))$.
By definition of $\Ff_1 \wr \Ff_2$, there is $f \in V_1^{H_2}$ and $u \in V_2$, such that $\phi(v) = (f,u)$.
Similarly, $\phi(f_i) = (h_i, \pi^{(2)}\phi(f_i))$ for some $h_i \in H_1$, for $i = 1,2$.
Thus,
\begin{equation}\label{eq:d-1}
\begin{split}
\phi(v(f_1+f_2)) &= (f,u) ((h_1^1, h_1^2) + (h_2^1, h_2^2)) \\
& = (f,u) (h_1+h_2, \pi^{(2)}\phi(f_1+f_2)) \\
& = (f(h_0)(h_1+h_2), u h_0) 
\end{split}
\end{equation}
On the other hand,
\begin{equation}\label{eq:d-2}
\begin{split}
\phi(vf_1+vf_2) &= (f,u) (h_1, h_0) + (f,u) (h_1,h_0) \\
&= (f(h_0)h_1, u h_0)+(f(h_0)h_2, u h_0) \\
& = (f(h_0)h_1+f(h_0)h_2, u h_0)
\end{split}
\end{equation}
Given that $\Ff_1$ is distributive, the last lines of (\ref{eq:d-1}) and (\ref{eq:d-2}) are the same.
\end{proof}

This statement has a converse, which can be shown using the Local-Global Theorem~\ref{thm:loc-glob}:

\begin{theorem}\label{prop:2-wreath}
A forest algebra $\Ff$ (finite or infinite) is 2-distributive if and only if it divides a wreath product of two (possibly infinite) distributive forest algebras.
\end{theorem}

\begin{proof}[Proof Sketch]
The `if' direction is the previous proposition.
For the `only if' direction, the proof closely follows that of Theorem~\ref{thm:main-result}.
Let $\sim$ be the congruence on $\Sigma^\Delta$ induced by $v[h+h'] = vh+vh'$.
The quotient of $\Sigma^\Delta$ by $\sim$ is an infinite distributive forest algebra, which we denote $\Sigma^\Delta_D$. 
We can extend $\pi$ to a morphism $\Sigma^\Delta \rightarrow \Sigma^\Delta_D$. 
Let $\phi : \Sigma^\Delta \rightarrow \Ff$ be any morphism.
Consider $D_{\phi, \pi}$ (Definition~\ref{def:derived-category}), an infinite forest category.
From the definition of 2-distributivity, one can show that $D_{\phi, \pi}$ is locally-distributive (Definition~\ref{def:loc-dist}).
Similar to Theorem~\ref{thm:loc-glob}, one can then show that $D_{\phi, \pi}$ divides an infinite distributive forest algebra $\Sigma'^\Delta_D$ where $\Sigma'$ is an extended alphabet.
The Derived Category Theorem~\ref{thm:derived-cat} implies that $\Ff$ divides $\Sigma'^\Delta_D \wr \Sigma^\Delta_D$.

\end{proof}

Allowing infinite algebras is crucial here:
Even if $\Ff$ is finite, we cannot readily conclude that it divides a product of two \emph{finite} distributive forest algebras.
Nonetheless, this characterization is interesting:
It shows that 2-distributive algebras represent the second level in the hierarchy of iterated wreath products of infinite distributive algebras.
This hierarchy can be viewed as an infinitary counterpart to PDL, as it allows products of infinite algebras.

Using this characterization, we can give a simple example:

\begin{example}\label{ex:langs-basic}
Any distributive algebra is evidently also 2-distributive.
For a less trivial example, consider the language of forests satisfying the following conditions: 
(1) Each maximal path has the form $a^*b$ or $a^c$, (2) each b-node has a c-sibling, (3) each c-node has a b-sibling
  (see Figure~\ref{fig:ex-lang-basic}).
This language is not recognized by a distributive algebra.
It is not hard to show that it is recognized by a wreath product of two finite forest algebras, and thus, using Proposition~\ref{prop:wreath-2}, by a 2-distributive forest algebra.
\end{example}

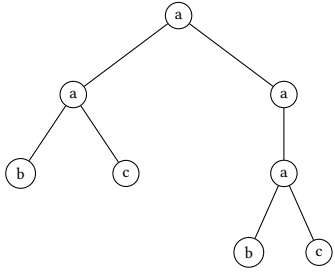
\begin{figure}
\begin{tikzpicture}[level/.style={sibling distance=40mm/#1}, scale=0.7,transform shape]
\node [circle,draw] (a) {a}
    child {node [circle,draw] (b) {a}
      child {node [circle, draw] {b}}
      child {node  [circle, draw] {c}}
    }
    child {node [circle,draw] (c) {a}
      child {node [circle,draw] (d) {a}
        child {node [circle, draw] {b}}
        child {node [circle, draw] {c}}
      }
    };
\end{tikzpicture}
\caption{Illustration for Example~\ref{ex:langs-basic}}\label{fig:ex-lang-basic}
\end{figure}

In Example~\ref{ex:langs-basic}, the recognizing algebra is not just 2-distributive but also divides the wreath product of two finite distributive algebras (the language is therefore in PDL).
However, this is not the case in general:
Finite 2-distributive algebras need not divide a wreath product of two \emph{finite} distributive forest algebras.
We will show this in the example below.
The main result of this paper will imply that a wreath product of \emph{four} finite distributive forest algebras will still be enough in this case.
This proves that all languages recognized by finite 2-distributive forest algebras are definable in PDL.

\begin{example}\label{ex:langs}
Consider the following languages (see Figure~\ref{fig:ex-langs} for illustration):

$\La_1$ is the set of nonempty forests where (1) each maximal path has the form $a^* b$ or $a^* c$, (2) each $b$-node has a $c$-sibling, (3) each $c$-node has a $b$-sibling (see Figure~\ref{fig:l1}).

$\La_2$ is the set of nonempty forests where (1) each maximal path has the form $a^* b$ or $a^+ c$, (2) each $b$-node has an $a$-sibling which has a $c$-child, (3) each $a$-node with a $c$-child has a $b$-sibling (see Figure~\ref{fig:l2}).

$\La^a_3$ is the set of (possibly empty) forests where each tree has the form $c_d(f_1 + f_2)$, where $c_d$ denotes the context consisting of only a node labeled $d$ and a variable below it, with $f_1 \in \La^b_3$, $f_2\in \La_1$.

$\La^b_3$ is the set of (possibly empty) forests where each tree has the form $c_d(f_1+f_2)$, with $f_1 \in \La^a_3$ and $f_2 \in \La_2$.

Set $\La := \La^a_3 + \La^b_3$  (see Figure~\ref{fig:l}).

It can be shown that $\La$ is recognized by a wreath product of two infinite distributive algebras and thus is 2-distributive by Proposition~\ref{prop:2-wreath}.
However, $\La$ is not recognized by any wreath product of two \emph{finite} distributive algebras.

The proof is based on facts about \emph{separation} by morphisms to distributive algebras:
If $\La$ is a language of forests, $\pi(\La) \subseteq Pow(\Sigma^*)$ is the image of $\La$ under $\pi$, a set of finite pathsets.
First, it is not hard to see that $\pi(\La_1) \cap \pi(\La_2)$ is empty, and thus the language $\pi^{-1}(\pi(\La_1))$ separates these: $\La_1 \subset \pi^{-1}(\pi(\La_1)) \subset (H_\Sigma - \La_2)$.
The syntactic forest algebra of $\pi^{-1}(\pi(\La_1))$ is distributive, but infinite (it crucially needs to count at arbitrary depths).
From this fact, one can derive using Theorem~\ref{prop:2-wreath} that $\La$ is indeed 2-distributive.

However, even though $\La_1$ and $\La_2$ are both regular, no language recognized by a \emph{finite} distributive algebra can separate them.
From this, one can deduce using the Derived Category Theorem~\ref{thm:derived-cat}  that $\La$ is not recognized by the wreath product of any two finite distributive forest algebras. 
However, it is not hard to show that $\La_1$ and $\La_2$ are both recognized by a wreath product of two finite distributive algebras. Fom this, one can derive that $\La$ is recognized by a wreath product of \emph{three} finite distributive forest algebras.
\end{example}

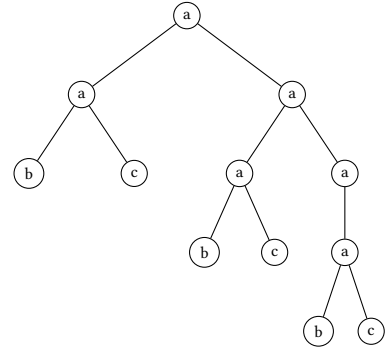
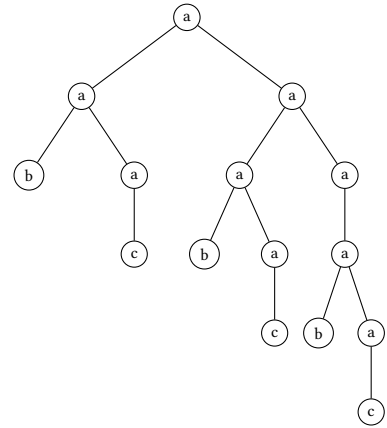
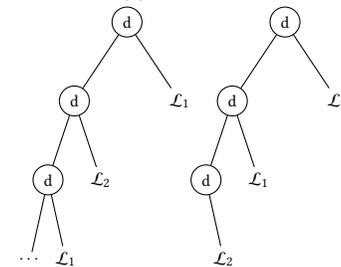
\begin{figure}
\begin{subfigure}[b]{0.3\textwidth}
\begin{tikzpicture}[level/.style={sibling distance=40mm/#1}, scale=0.7,transform shape]
\node [circle,draw] (a) {a}
    child {node [circle,draw] (b) {a}
      child {node [circle, draw] {b}}
      child {node  [circle, draw] {c}}
    }
    child {node [circle,draw] (c) {a}
      child {node [circle,draw] (d) {a}
        child {node [circle, draw] {b}}
        child {node [circle, draw] {c}}
      }
      child {node [circle, draw] {a}
             child {node [circle,draw] (e) {a}
               child {node [circle, draw] {b}}
               child {node [circle, draw] {c}}
             }
            }
    };
\end{tikzpicture}
\caption{An element of $\La_1$}\label{fig:l1}
\end{subfigure}
\begin{subfigure}[b]{0.3\textwidth}
\begin{tikzpicture}[level/.style={sibling distance=40mm/#1}, scale=0.7,transform shape]
\node [circle,draw] (a) {a}
    child {node [circle,draw] (b) {a}
      child {node [circle, draw] {b}}
      child {node  [circle, draw] {a} child {node  [circle, draw] {c}}   }
    }
    child {node [circle,draw] (c) {a}
      child {node [circle,draw] (d) {a}
        child {node [circle, draw] {b}}
        child {node  [circle, draw] {a} child {node  [circle, draw] {c}}   }
      }
      child {node [circle, draw] {a}
             child {node [circle,draw] (e) {a}
               child {node [circle, draw] {b}}
               child {node  [circle, draw] {a} child {node  [circle, draw] {c}}   }
             }
            }
    };
\end{tikzpicture}
\caption{An element of $\La_2$}\label{fig:l2}
\end{subfigure}
\begin{subfigure}[b]{0.3\textwidth}
\begin{tikzpicture}[level/.style={sibling distance=20mm/#1}, scale=0.7,transform shape]
\node at (0,0) [circle,draw] (1st) {d}
    child {node [circle,draw] (b) {d}
      child {node [circle, draw] {d}
        child {node  {$\dots$}} 
        child {node  {$\La_1$}}
       }
      child {node   {$\La_2$}}
    }
    child {node  {$\La_1$}};
\node at (3,0) [circle,draw] (a) {d}
    child {node [circle,draw] (b) {d}
      child {node [circle, draw] {d}
        child {node  {} edge from parent[draw=none]}
        child {node {$\La_2$}}            }
      child {node   {$\La_1$}}
    }
    child {node   {$\La_2$}};
\end{tikzpicture}
\caption{A schematic depiction of $\La$: The left tree belongs to $\La_3^a$, the right tree one to $\La_3^b$.}\label{fig:l}
\end{subfigure}
\caption{Illustration for Example~\ref{ex:langs}. $\La$ is 2-distributive and regular, but not recognized by the wreath product of two finite distributive forest algebras.}\label{fig:ex-langs}
\end{figure}


\section{The Derived Forest Category}\label{sec:categories}

The proof of our main result will construct wreath product decompositions by separately studying the left and right factors of a wreath product.
Given a 2-distributive forest algebra $(H_2, V_2)$, we will use the Separation Lemma~\ref{cor:separation} to construct an intended right-hand factor $(H_2, V_2)$ which is already known to be a wreath product of finite distributive forest algebras.
The remaining problem is then to find a left-hand factor $(H,V)$ such that
%
%
%
%
%
%
\begin{equation}\label{eq:wreath-eq}
(H_1,V_1) \prec (H,V) \wr (H_2,V_2)
\end{equation}
holds. If we can show that this factor can be chosen to be distributive, the problem is solved.
In order to do this, we seek a general strategy to obtain a `minimal' left-hand factor $(H,V)$.
%
In the case of groups, the solution to this problem is provided by the kernel group, $\operatorname{ker}\ \phi$, when $\phi : G\rightarrow H$: Then $G$ is embedded in $\operatorname{ker} \phi \wr H$.

For monoids, the analog to $\operatorname{ker} \phi$ is not a monoid any more, but a category.
This classical construction is known as the \emph{Derived Category}~(\cite{tilson-categories-1987}, \cite{rhodes-q-theory-2009}) and originated from the study of regular word languages via wreath products of finite monoids \cite{brzozowski-characterizations-1973}, \cite{therien-graph-1985}, \cite{straubing-finite-1985}.
Recently, \cite{straubing-forest-2018} showed that this construction generalizes to the setting of forest algebras.

The sense in which this construction is `optimal' is made precise in Tilson's Derived Category Theorem~\cite{tilson-categories-1987}.
In the case of forest algebras, we will essentially see that the decomposition in (\ref{eq:wreath-eq})
holds if and only if $(H,V)$ is divided by a certain forest category determined from $(H_1,V_1)$ and $(H_2,V_2)$ and morphisms from $\Sigma^\Delta$ into these, where the notion of `division' by a category will be made precise below.

In this section, we will review the definition of the Derived Forest Category and the Derived Category Theorem from \cite{straubing-forest-2018}.
The Derived Forest Category is a category with some structure added to it, and the relation between forest categories and categories is akin to the relation between forest algebras and monoids.
It is possible to define a general notion of Forest Categories \cite{straubing-forest-2018}, but for our purposes, it is sufficient to consider the Derived Forest Category:


%
%
%
%
%
%

\begin{defin}\label{def:derived-category}[Derived Forest Category, \cite{straubing-forest-2018}]
Let $\Sigma$ be a finite alphabet. Consider a pair of surjective forest algebra morphisms
$$(H_1, V_1) \xleftarrow{\alpha} \Sigma^\Delta \xrightarrow{\beta} (H_2, V_2)$$
The derived category $D_{\alpha,\beta}$ is defined as follows:
\begin{enumerate}
\item The set of \emph{objects} of the category is $\Obj(D_{\alpha,\beta}) := H_2$
\item As in ordinary categories, \emph{arrows} connect objects.
To define the arrows, we fix $h, h' \in H_2$ and introduce an equivalence relation on the set of triples $(h,p,h')$ with $p \in V_\Sigma$ for which $\beta(p)\cdot h=h'$. We set $(h,p,h')\sim(h,q,h')$ if for all $s\in H_ \Sigma$ with $\beta(s)=h$, we have $\alpha(ps) = \alpha(qs)$.

We then set $\Arr(h,h')$ to be the set of equivalence classes of $\sim$.
Its elements are called \emph{arrows}.

We depict an arrow as $\carrow{h}{p}{h'}$, with the understanding that the same arrow can have many distrinct representations in this form.

For consistency with notation here, the direction of arrows in this graphical notation is inverted relative to \cite{straubing-forest-2018}.


\item To obtain a category, we now want to define the multiplication of arrows. We set
$$\left( \carrow{h_2}{p}{h_3}\right)\cdot\left(\carrow{h_1}{q}{h_2}\right)=\carrow{h_1}{pq}{h_3}$$
or shortened
$$\carrow{\carrow{h_1}{q}{h_2}}{p}{h_3}=\carrow{h_1}{pq}{h_3}$$
It can be shown that this is a well-defined arrow, independently of the representation chosen for the arrows \cite{straubing-forest-2018}.
Since the multiplication on $V_\Sigma$ is associative, this multiplication is associative.
The arrow $\carrow{h}{1_{V_\Sigma}}{h}$ is the identity at $h\in H_2$.


Up to this point, we have defined a category. We now add some additional structure:

\item For $h \in H_2$, we set
$$\HArr(h) := \{(\alpha(s), h) : s \in H_\Sigma, \beta(s) = h\}$$
The elements of this set are called \emph{half-arrows}. They can be thought of as being arrows that end in an object, but do not start in any object.

We depict the half-arrow $(h_1, h_2)$ as $\carrow{}{h_1}{h_2}$.

\item We set $\HArr(D_{\alpha,\beta})$ to be the set of all half-arrows. Note that $\Obj(D_{\alpha,\beta})$ and $\HArr(D_{\alpha,\beta})$ are monoids, and that the projection of a half-arrow onto the second element (that is, the object at the end of the arrow in our graphical notation) is a homomorphism from $\HArr(D_{\alpha,\beta})$ to $\Obj(D_{\alpha,\beta})$.

\item Viewed in analogy to forest algebras, arrows correspond to context types, while half-arrows correspond to forest types. We therefore want arrows to act on half-arrows. We define the action of an arrow on a half-arrow by
$$\left(\carrow{h_2}{p}{h_2'}\right)\cdot\left(\carrow{}{h_1}{h_2}\right)=\carrow{}{\alpha(p)h_1}{h_2'}$$
or shortened
$$\carrow{\carrow{}{h_1}{h_2}}{p}{h_2'}=\carrow{}{\alpha(p)h_1}{h_2'}$$
\item In analogy to forest algebras, we want to be able to add arrows and half-arrows. We set
$$\left(\carrow{h}{p}{h'}\right)+\left(\carrow{}{h_1}{h_2}\right) = \carrow{h}{p+s}{\left(h'+h_2\right)}$$
where $s\in H_\Sigma$ such that $\alpha(s) = h_1$, $\beta(s)=h_2$.
\end{enumerate}
For proofs of well-definedness, we refer the reader to  \cite{straubing-forest-2018}.
\end{defin}


To formulate the Derived Category Theorem connecting Derived Categories with wreath products, it is necessary to generalize the notion of division to the setting of forest categories dividing forest algebras:

\begin{defin}\label{def:division}[Division, \cite{straubing-forest-2018}]
If $C$ is a derived forest category and $(H,V)$ a forest algebra, then we write $C\prec(H,V)$, and say $C$ \emph{divides} $(H,V)$, if for each $\left(\carrow{}{c}{x}\right) \in \HArr(C)$ there exists a nonempty set $K_{\carrow{}{c}x} \subseteq H$, and for each $\left(\carrow{x}{d}{y}\right)\in \Arr(C)$ there exists a nonempty set $K_{\carrow{x}{d}{y}}\subseteq V$ satisfying the following properties:
\begin{enumerate}
\item (Preservation of Operations) For all $\carrow{}{c}x$, $\carrow{}{d}y \in \HArr(C)$, $\carrow{x}{e}{y}, \carrow{y}{f}{z} \in \Arr(C)$,
\begin{enumerate}
\item $K_{\carrow{y}{f}{z}} \cdot K_{\carrow{x}{e}{y}} \subseteq K_{\carrow{\carrow{x}{e}{y}}{f}{z}}$
\item $K_{\carrow{x}{e}y} \cdot K_{\carrow{}{c} x} \subseteq K_{\carrow{\carrow{}{c}{x}}{e}{y}}$
\item $K_{\carrow{}{c}x}+K_{\carrow{}{d}y}\subseteq K_{\carrow{}{c}{x}+\carrow{}{d}{y}}$
\item $K_{\carrow{}{c}x} + K_{\carrow{y}{f}z} \subseteq K_{\carrow{}{c}x+y\carrow{}{f}z}$
\item $K_{\carrow{y}{f}z} + K_{\carrow{}{c}x} \subseteq K_{\carrow{y}{f}z+\carrow{}{c}x}$
\end{enumerate}
\item (Injectivity)
\begin{enumerate}
\item If $\carrow{x}{c}y$ and $\carrow{x}{c'}y$ are distinct arrows, then $K_{\carrow{x}{c}y}\cap K_{\carrow{x}{c'}y} = \emptyset$
\item If $\carrow{}{c}y$ and $\carrow{}{c'}y$ are distinct half-arrows, then $K_{\carrow{}{c}y}\cap K_{\carrow{}{c'}y}=\emptyset$.
\end{enumerate}
\end{enumerate}

\end{defin}

In the special case where a derived forest category has exactly one object, it can be viewed as a forest algebra.
In this case, the notion of division reduces to ordinary division of forest algebras.

We are now ready to state the Derived Category Theorem connecting categories with wreath products.
We state only the direction required for our main result:

\begin{theorem}\label{thm:derived-cat}[Derived Category Theorem,  \cite{straubing-forest-2018}]
Let $\Sigma$ be an alphabet, and let $\alpha, \beta$ be morphisms from $\Sigma^\Delta$ onto forest algebras $(H_1, V_1)$, $(H_2, V_2)$, respectively. Let $(H,V)$ be a forest algebra.
Assume $D_{\alpha, \beta} \prec (H,V)$. Then
$$(H_1,V_1) \prec (H,V) \wr (H_2, V_2)$$
\end{theorem}


\section{Main Result}\label{sec:main-result}

Our aim is to prove that any language recognized by a finite 2-distributive forest algebras is definable in PDL:

\begin{theorem}\label{thm:main-result}
Let $\Ff$ be a finite 2-distributive forest algebra.
Then every language recognized by $\Ff$ is definable in PDL.
\end{theorem}
In the Discussion (Section~\ref{sec:discussion}), we will discuss the relevance of this result for the question of deciding definability in PDL.

Our proof proceeds by solving two sub-problems related to the left and right factors in wreath product decompositions:
To obtain the right factor of a wreath product decomposition, we study the problem of \emph{separating} forest languages by the map $\pi$.
To then obtain the left factor, we start at the Derived Category, and show that it has a certain local property -- in our case, local distributivity.
To conclude a decomposition result, we then prove that this local property entails a global property.
These steps are remarkably similar to results from the theory of logic on words and finite monoids which also reduce the problem of decidability to separation \cite{place-going-2014} and Local-Global theorems \cite{krebs-effective-2012}.


\subsection{Separation Lemma}
We will first state the Separation Lemma.
Recall that $\pi(f)$ is the set of paths in the forest $f$.
If $\La$ is a language of forests, $\pi(\La) \subseteq Pow(\Sigma^*)$ is the image of $\La$ under $\pi$, a set of finite pathsets.

\begin{llemma}[Separation Lemma]\label{cor:separation}
Let $\La_1, \La_2 \subseteq H_\Sigma$ be regular forest languages such that
$$\pi(\La_1) \cap \pi(\La_2) = \emptyset$$
Then there are finite distributive algebras $\Ff_1, \Ff_2, \Ff_3$ and a language $X \subseteq \Sigma^\Delta$ recognized by $\Ff_1\wr\Ff_2\wr\Ff_3$ such that $$\La_1 \subseteq X \subseteq (H_\Sigma - \La_2)$$
That is, $X$ \emph{separates} $\La_1$ from $\La_2$.
\end{llemma}

\begin{proof}
The proof considers a forest algebra $(H,V)$ recognizing both $\La_1$ and $\La_2$ via a morphism $\phi$, and constructs PDL languages `approximating' each $\phi^{-1}(h)$ for $h \in H$.
\end{proof}

Let's consider how this is useful for proving Theorem~\ref{thm:main-result} by sketching the proof idea for the theorem -- we will make this more precise in Section~\ref{rect:result-proof}.
If $\Ff$ is a forest algebra with morphism $\phi_\Ff : \Sigma^\Delta \rightarrow \Ff$, we can apply this lemma to all pairs of languages $\phi_\Ff^{-1}(h)$ for $h \in H_\Ff$ and obtain separators $X_{h,h'}$ for each pair whenever $\pi\phi^{-1}(h)\cap\pi\phi^{-1}(h') = \emptyset$.
Combining the resulting separators, we can build finite distributive algebras $\Ff_1, \Ff_2, \Ff_3$ and a morphism $\phi_X : \Sigma^\Delta \rightarrow \Ff_1\wr\Ff_2\wr\Ff_3$ which recognizes each $X_{h,h'}$.
We will examine the derived category $D_{\phi_\Ff, \phi_X}$.
If we can show that this derived category divides a finite distributive algebra $\Ff_0$, we can apply the Derived Category Theorem to conclude $\Ff \prec \Ff_0 \wr \Ff_1 \wr \Ff_2 \wr \Ff_3$, which then will prove Theorem~\ref{thm:main-result}.

\subsection{Locally-Distributive Categories}
Recall the equation $v [h_1 + h_2] = v h_1 + v h_2$ defining distributive forest algebras.
To apply this to derived forest categories, we would want to interpret $v$ as an arrow and $h_1, h_2$ as half-arrows.
In general, this equation does not make sense, since the action of an arrow on a half-arrow is only defined in certain cases. 
The equation becomes sensible when we restrict it to those arrows and half-arrows for which the action is defined:

\begin{defin}[Locally Distributive]\label{def:loc-dist}
We say that a derived forest category $C$ is \emph{locally distributive} if the following equation holds
for any $h \in \Obj(C)$ and any two half-arrows $h_1, h_2 \in \HArr(h)$, and for any arrow $v \in \Arr(h,h')$ ($h' \in \Obj(C)$):
$$v [h_1 + h_2] = v h_1 + v h_2$$
Rewriting this in arrow-based notation, we want the following for any half-arrows $\carrow{}{h_1}h$ and $\carrow{}{h_2}h$, and for any arrow $\carrow{h}{v}{h'}$:
$$\left(\carrow{h}{v}{h'}\right)\cdot\left(\carrow{}{h_1}h + \carrow{}{h_2}h\right) = \left(\carrow{\carrow{}{h_1}{h}}{v}{h'}\right) + \left(\carrow{\carrow{}{h_2}{h}}{v}{h'}\right)$$
\end{defin}

Any derived forest category that divides a distributive forest algebra must be locally distributive.
We now show that the converse is also true: Any locally-distributive category divides some distributive forest algebra.
In analogy to results from the theory of ordinary finite categories, we refer to this as a Local-Global Theorem -- showing that being locally distributive entails a global property of the category:

\begin{theorem}[Local-Global]\label{thm:loc-glob}
Let $C$ be a locally-distributive finite derived forest category.
Then it divides a finite distributive forest algebra.
\end{theorem}

\begin{proof}
The proof proceeds by considering terms built from arrows and half-arrows and using local distributivity to transform them into a normal form that only depends on the paths in these terms (viewing them as forests). 
We then apply Proposition~\ref{prop:distr-char} to construct a finite distributive forest algebra and a division.
\end{proof}

The idea of introducing a `local' version of distributivity that is appropriate for forest categories, and then relating it to distributive forest algebras in a `Local-Global' Theorem is related to a long tradition in the theory of semigroups and monoids, where local pseudovarieties of categories have been an important object of study (e.g., \cite{tilson-categories-1987}, \cite{pin-locally-1988}, \cite{almeida-syntactical-1996}), and where such Local-Global theorems have been applied to prove decidability of logic classes~\cite{krebs-effective-2012}.

In order to prove Theorem~\ref{thm:main-result}, our goal will be to prove that the derived category $D_{\phi_\Ff, \phi_X}$ mentioned above is locally distributive, then being able to apply Theorem~\ref{thm:loc-glob}.
The details are given in Section \ref{rect:result-proof}.

\subsection{Concluding the proof of Theorem~\ref{thm:main-result}}\label{rect:result-proof}

\begin{proof}[Proof of the Theorem]
Let $\La$ be a language recognized by a 2-distributive finite algebra $\Ff$ via morphism $\phi : \Sigma^\Delta \rightarrow \Ff$.
For each pair $h_1, h_2 \in H_\Ff$ such that $\pi(\phi^{-1}(h_1)) \cap \pi(\phi^{-1}(h_2)) = \emptyset$, we can apply Lemma~\ref{cor:separation} to the languages $\La_1 := \phi^{-1}(h_1)$ and $\La_2 := \phi^{-1}(h_2)$.
From the lemma we get a language $X_{h_1,h_2}$ separating the preimages of $h_1$ and $h_2$.
That is, we have
$$\phi^{-1}(h_1) \subseteq X_{h_1,h_2} \subseteq (H_\Sigma - \phi^{-1}(h_2))$$
We also get an algebra $\Gg_{h_1, h_2} = \Ff_1^{h_1,h_2} \wr \Ff_2^{h_1,h_2} \wr \Ff_3^{h_1,h_2}$ which recognizes $X$ via morphism $\psi_{h_1, h_2} : \Sigma^\Delta \rightarrow \Gg_{h_1, h_2}$.
By the lemma, the three algebras $\Ff_1^{h_1,h_2}, \Ff_2^{h_1,h_2}, \Ff_3^{h_1,h_2}$ are finite and distributive.

Let $$\Gg := \left(\bigtimes_{h_1, h_2} \Ff_1^{h_1, h_2}\right) \wr \left(\bigtimes_{h_1, h_2} \Ff_2^{h_1, h_2}\right) \wr \left(\bigtimes_{h_1, h_2} \Ff_3^{h_1, h_2}\right)$$
where we take products over all pairs $h_1, h_2$ for which $\pi(\phi^{-1}(h_1)) \cap \pi(\phi^{-1}(h_2)) = \emptyset$ holds.
Then $\Gg$ is a wreath product of three finite distributive algebras.
Furthermore, it is divided by each $\Gg_{h_1, h_2}$, and thus recognizes each of the separators $X_{h_1, h_2}$ via some morphism $\psi : \Sigma^\Delta \rightarrow \Gg$.

We now consider the derived forest category $D_{\phi,\psi}$.
We want to show that it is locally distributive, then being able to apply the Derived Category Theorem.
Let $h, h'$ be objects, let $f_1, f_2 \in \HArr(h)$, and let $v \in \Arr(h,h')$.
By the definition of the Derived Category, we can write $f_1$ as $\carrow{}{h_1}h$ and $f_2$ as $\carrow{}{h_2}h$.
Also, we can write $v$ as $\carrow{h}{p} h'$, with $p \in V_\Ff$.
We want to prove the equality from Definition~\ref{def:loc-dist}.

Note that $h_1, h_2 \in H_\Ff$.
In view of the construction of the half-arrows in the derived category, there are forests $t_1, t_2 \in H_\Sigma$ such that $\phi(t_i) = h_i$ and $\psi(t_i) = h$ for $i=1,2$.
For a contradiction, let us assume $\pi(\phi^{-1}(h_1)) \cap \pi(\phi^{-1}(h_2)) = \emptyset$.
Given the way $\psi$ was constructed, the equality $\psi(t_i) = h$ entails $\psi_{h_1, h_2}(t_1) =\psi_{h_1, h_2}(t_2)$.
This is a contradiction to the way in which we have chosen $\psi_{h_1,h_2}$.
The assumption about the empty intersection must have been incorrect, and we have
$$\pi(\phi^{-1}(h_1)) \cap \pi(\phi^{-1}(h_2)) \neq \emptyset$$
So there are forests $b_1, b_2$ such that $\phi(b_i) = h_i$ and $\pi(b_1) = \pi(b_2)$.

Recall $v = \carrow{h}{p} h'$, with $p \in V_\Ff$.
Let $\alpha' \in \phi^{-1}(p)$.
Since $\Ff$ is 2-distributive, we have $\phi(\alpha'[b_1+b_2]) = \phi(\alpha' b_1 + \alpha' b_2)$.
Applying $\phi$, this means $$p(h_1+h_2) = p(h_1) + p(h_2)$$
In the derived category, this translates to $$v (f_1 + f_2)) = v f_1 + v f_2$$
or, in arrow-based notation,
$$\carrow{\left(\carrow{}{h_1}h + \carrow{}{h_2}h\right)}{p}{h'} = \left( \carrow{\carrow{}{h_1}h}{p}{h'}\right) + \left( \carrow{\carrow{}{h_2}h}{p}{h'}\right)$$
Thus, $D_{\phi,\psi}$ is locally distributive.
It is also finite (the two algebras involved in its construction are finite), so it divides a finite distributive algebra $\Ff'$.
By the Derived Category Theorem, $\La$ is recognized by $\Ff' \wr \Gg$, which is the wreath product of four finite distributive algebras.
\end{proof}

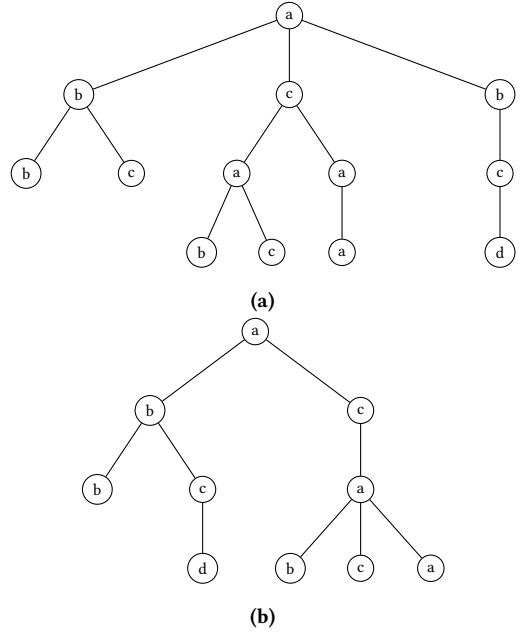
\begin{figure}
\begin{subfigure}[b]{0.4\textwidth}
\begin{center}
\begin{tikzpicture}[level/.style={sibling distance=40mm/#1}, scale=0.7,transform shape]
\node [circle,draw] (a) {a}
    child {node [circle,draw] (b) {b}
      child {node [circle, draw] {b}}
      child {node  [circle, draw] {c}}
    }
    child {node [circle,draw] (c) {c}
      child {node [circle,draw] (d) {a}
        child {node [circle, draw] {b}}
        child {node [circle, draw] {c}}
      }
      child {node [circle, draw] {a}
             child {node [circle,draw] (e) {a}
             }
            }
    }
    child {node [circle,draw] (b) {b}
      child {node  [circle, draw] {c}
        child {node [circle, draw] {d}}
            }
    }
;
\end{tikzpicture}
\end{center}
\caption{}
\end{subfigure}
\begin{subfigure}[b]{0.4\textwidth}
\begin{center}
\begin{tikzpicture}[level/.style={sibling distance=40mm/#1}, scale=0.7,transform shape]
\node [circle,draw] (a) {a}
    child {node [circle,draw] (b) {b}
      child {node [circle, draw] {b}}
      child {node  [circle, draw] {c}
        child {node [circle, draw] {d}}
            }
    }
    child {node [circle,draw] (c) {c}
      child {node [circle, draw] {a}
        child {node [circle, draw] {b}}
        child {node [circle, draw] {c}}
             child {node [circle,draw] (e) {a}
             }
            }
    }
;
\end{tikzpicture}
\end{center}
\caption{}
\end{subfigure}
\caption{Illustration for Definition~\ref{def:psi}. Applying $\Psi$ to the tree in (a) results in the tree in (b). The trees have the same (not necessarily maximal) paths. In (b), no two siblings have the same label.}\label{fig:psi}
\end{figure}

\section{Decidability of 2-Distributivity}

We have shown that languages recognized by finite 2-distributive algebras form a subclass of PDL.
We now show that 2-distributivity is decidable.


To decide whether $\Ff$ is 2-distributive, we need to check for morphisms $\phi : \Sigma^\Delta\rightarrow\Ff$ whether $\phi(v(f_1+f_2)) = \phi(vf_1 + vf_2)$ holds whenever $\pi(f_1) = \pi(f_2)$, for all forests $f_1, f_2 \in H_\Sigma$ and all contexts $v \in V_\Sigma$.
To do this algorithmically, we want to find those pairs $h, h' \in H$ such that $\pi(\phi^{-1}(h))$ and $\pi(\phi^{-1}(h'))$ have nonempty intersection. For these $h, h'$, we then need to check whether $v(h+h') = vh+vh'$ for all $v \in V$. If we can find these pairs $h, h'$ algorithmically, decidability is shown (We will see that looking at one specific morphism $\phi$ is enough.).

Thus, the problem boils down to deciding, given two regular forest languages $\La_1, \La_2$, whether $\pi(\La_1) \cap \pi(\La_2)$ is empty.
We will reduce this to the problem of deciding whether two regular forest languages -- computed from $\La_1, \La_2$ -- have nonempty intersection.
We will use the following tool: 
\begin{defin}[Distributive Normal Form]\label{def:psi}
Define a map $\Psi : H_\Sigma \rightarrow H_\Sigma$ as follows.

Consider a forest $f := \beta_1[f_1] + ... + \beta_n[f_n]$ ($n \geq 0$).
Here, $\beta_1, ..., \beta_n$ are symbols from $\Sigma$, some or all of which can be identical, and $f_1, ..., f_n$ are forests.
For each $\beta \in \{\beta_1, .., \beta_n\}$, define \[F_\beta := \{f_i : \beta_i = \beta\} \subset \{f_1, ..., f_n\}\]
Then, set $$\Psi(f) := \sum_{\beta \in \{\beta_1, ..., \beta_n\}} \beta\left(\Psi\left[\sum_{f' \in F_\beta} f'\right]\right)$$
\end{defin}

An example is provided in Figure~\ref{fig:psi}.

\begin{prop}\label{prop:psi}
Let $f, f' \in H_\Sigma$.
\begin{enumerate}
\item No two distinct sibling nodes in $\Psi(f)$ are labeled with the same symbol.

\item $\pi(f) = \pi(\Psi(f))$.

\item $\pi(f) = \pi(f')$ if and only if $\Psi(f) = \Psi(f')$.
\end{enumerate}
\end{prop}

\begin{proof}
(1) By induction over the height of forests.
(2) is immediate from the definition of $\Psi$.
For (3), the `if' direction follows from (2). For the `only if' direction, observe that, for any given pathset, there is only a single forest (up to order of siblings) having this pathset and satisfying the condition that no sibling nodes have the same symbol.
\end{proof}

Due to the second property, we will use $\Psi(f)$ as a suitable representative forest for the path set $\pi(f)$.
In certain ways, $\Psi(f)$ will be better-behaved than a general forest $f$.
The following proposition shows that the image of languages under $\Psi$ is also well-behaved:

\begin{prop}\label{prop:psi-reg}
Let $\La$ be a regular forest language. Then $\Psi(\La)$ is a regular forest language and can be effectively constructed from (an automaton for) $\La$.
\end{prop}

It is important to note that the image $\Psi(\La)$ is not recognized by a horizontally idempotent forest algebra, as multiplicity of children does matter. The recognizing finite forest algebra will be horizontally commutative but not idempotent.
This proposition and proof is the one place in this paper where we deviate from our convention that all forest algebras are horizontally commutative and idempotent.

\begin{proof}
Let $\Ff = (H,V)$ be a finite forest algebra recognizing $\La$ via morphism $\phi$.

Let $H' := {Pow(Pow(H))}^\Sigma \cup \{\bot\}$ with the operation: $f + f' = \bot$ if there is $\alpha \in \Sigma$ such that $f(\alpha), f'(\alpha) \neq \emptyset$, or if one of $f, f'$ is already equal to $\bot$.
Otherwise, $(f+f')(\alpha) = f(\alpha) \cup f'(\alpha)$.
It should be noted that this operation is not idempotent due to the first condition, which makes $f + f = \bot$ unless $f(\alpha) = \emptyset$ for all $\alpha$.
With this operation, $H'$ is a commutative finite monoid with identity $f_0$ given by $f_0(\alpha) = \emptyset$ for $\alpha \in \Sigma$.

Let $V'$ be the finite monoid of all maps $H' \rightarrow H'$, which naturally acts on $H'$.
It is easy to show that $(H',V')$ is a finite forest algebra (though not horizontally idempotent).

Then define a morphism $\psi : \Sigma^\Delta \rightarrow (H',V')$ by first constructing the images of the contexts consisting of only a single letter: $\psi(\alpha) := g_\alpha$ where by $\psi(\alpha)$ we denote the image of the context consisting of only $\alpha$ and a variable below it ($\alpha[X]$).
Once we have chosen $g_\alpha \in V'$ for each $\alpha$, it is not hard to see that we obtain a unique forest algebra morphism $\psi : \Sigma^\Delta \rightarrow (H',V')$ extending this map.
Recall that $g_\alpha$ will need to be a map $H' \rightarrow H'$.
We first set $g_\alpha(\bot) = \bot$ for all $\alpha$.
For $f \in H' - \{\bot\}$, so $f : \Sigma \rightarrow Pow(Pow(H))$, we furthermore set
\[\psi(\alpha)(f)(\beta) = \emptyset \text{ when }\alpha \neq \beta\]
Finally, considering the case $\alpha = \beta$, then for any $Q \subset H$, we set $Q \in \psi(\alpha)(f)(\alpha)$ if and only if
there are sets $Q_1, ..., Q_l \subset H$ such that for each $\gamma \in \Sigma$ such that $f(\gamma) \neq \emptyset$, there is $P_\gamma \in f(\gamma)$ such that
\[Q_1 \cup ... \cup Q_l = \bigcup_{\gamma} P_\gamma\]
 and 
\[Q = \left\{\phi(\alpha)\cdot\left[\sum_{h \in Q_i} h\right] : i=1, ..., l\right\}\]
%
%
%
We now claim that, for $h \in H$, the language $\Psi(\phi^{-1}(h)) - \{\emptyset\}$ (that is, removing the empty forest if it is in the language) is equal to 
\[\psi^{-1}(\{f : \exists \alpha : f(\alpha) \neq \emptyset \wedge \forall \alpha \in \Sigma : f(\alpha) = \emptyset \vee \{h\} \in f(\alpha)\})\]
where $\psi$ is the forest algebra morphism $\psi : \Sigma^\Delta \rightarrow (H',V')$ that we just constructed.
This is shown by induction over forests.

Considering that the empty forest is the only element of the preimage of the identity element of $H'$, this implies that $\Psi(\La)$ is recognized by $(H',V')$ via $\psi$.
\end{proof}

We can now show decidability of 2-distributivity:
\begin{theorem}\label{prop:2-decid}
It is decidable whether a finite forest algebra is 2-distributive.
\end{theorem}

\begin{proof}
Given a finite forest algebra $\Ff = (H,V)$, choose $\Sigma := V$, and let $\phi : \Sigma^\Delta \rightarrow (H,V)$ be the (unique) morphism extending the identity map $\phi : \Sigma \rightarrow V$, that is, mapping the context $v[X]$ to $v \in V$.

Given two regular forest languages $\La_1, \La_2$, it is decidable whether $\pi(\La_1) \cap \pi(\La_2)$ is empty.
To prove this, we use the mapping $\Psi$ introduced in Definition~\ref{def:psi}.
We can use Proposition~\ref{prop:psi-reg} to effectively check whether the regular forest languages $\Psi(\La_1)$ and $\Psi(\La_2)$ have nonempty intersection.
From Proposition~\ref{prop:psi}, we know that this happens if and only if $\pi(\phi^{-1}(h))$ and $\pi(\phi^{-1}(h'))$ also have nonempty intersection.

Using this resut, we can then, for each pair $h, h' \in H$, effectively check whether $\pi(\phi^{-1}(h))$ and $\pi(\phi^{-1}(h'))$ have nonempty intersection.
If this is the case, we can check for each context type $v \in V$ whether $v[h+h'] = vh+vh'$.
This equality holds for each $v$ and for each selected pair $h, h'$ if and only if $\phi(c(f+f')) = \phi(cf+cf')$ for all $c \in V_\Sigma$ and each $f, f' \in H_\Sigma$ such that $\pi(f) = \pi(f')$.
This is a necessary condition for $\Ff$ to be 2-distributive.

To prove that this is also sufficient, consider another alphabet $\Sigma'$ and a morphism $\psi : \Sigma'^\Delta \rightarrow (H,V)$.
We can build a morphism $\eta : \Sigma'^\Delta \rightarrow \Sigma^\Delta$, generated by the map $\Sigma' \rightarrow \Sigma$ defined by $\eta(\alpha) := \psi(\alpha) \in \Sigma$ for $\alpha \in \Sigma'$.
Then $\psi = \phi \circ \eta$.
Let $f_1, f_2 \in H_{\Sigma'}$ with $\pi(f_1) = \pi(f_2)$, and let $c \in V_{\Sigma'}$.
Then $\pi(\eta(f_1)) = \pi(\eta(f_2))$, and, by assumption, $\psi(c[f_1+f_2]) = \phi(\eta(c[f_1+f_2])) = \phi(\eta(c)[\eta(f_1)+\eta(f_2)]) = \phi(\eta(c)\eta(f_1)+\eta(c)\eta(f_2)) = \phi(\eta(cf_1+cf_2)) = \psi(cf_1+cf_2)$.

\end{proof}

\section{Discussion}\label{sec:discussion}

We have shown that 2-distributive finite forest algebras recognize a subclass of PDL, and that 2-distributivity is a decidable property.

As we outlined in the Introduction, generalizing this approach to $k > 2$ would settle decidability of PDL.
Our notion of 2-distributivity can be generalized in the following way, slightly different than the one given in~\cite{straubing-new-2013}:

\begin{defin}
For each $k \geq 1$, define a congruence $\sim_k$ on $\Sigma^\Delta = (H_\Sigma, V_\Sigma)$ as follows:
\begin{enumerate}
\item $\sim_1$ is the smallest congruence such that $f \sim_1 f'$ whenever $f, f' \in H_\Sigma$ and $\pi(f) = \pi(f')$.
\item For any $k \geq 1$, $\sim_{k+1}$ is the smallest congruence such that \[v[f+f'] \sim_{k+1} vf + vf'\] for any $v \in V_\Sigma$ and any $f, f' \in H_\Sigma$ such that $f \sim_k f'$.
\end{enumerate}
\end{defin}
For each $k$, the congruence $\sim_k$ encodes a $k$-fold iteration of the distributive law.
A forest algebra $\Ff$ is \emph{$k$-distributive} if, for all morphisms $\phi : \Sigma^\Delta \rightarrow \Ff$, $\phi(f) = \phi(f')$ whenever $f \sim_k f'$.
For $k = 2$, this coincides with our definition above.

In analogy to Proposition~\ref{prop:wreath-2} and a result from ~\cite{straubing-new-2013}, it can be shown that the wreath product of $k$ distributive forest algebras is k-distributive.
Determining, given a finite forest algebra $\Ff$, whether it is $k$-distributive for some $k$ is a decidable problem.
If one could show that any finite $k$-distributive forest algebra only recognizes languages in PDL, definability in PDL would therefore be shown decidable~\cite{straubing-new-2013}.
We have solved this problem in the case $k=2$.

Indeed, generalizing our Local-Global Theorem to $k \geq 2$ is feasible, and the proof method of our main result might be adapted to construct an inductive proof.
It would be sufficient to, given a general $k$-distributive algebra, construct a wreath product of finite distributive algebras and show that an appropriate derived forest algebra is $k-1$-distributive. 
To carry this out, a suitable strengthening of our Separation Lemma to a property stronger than separation would be required.

PDL is a member of a larger family of forest languages for which decidability of expressibility is still unknown, in spite of longstanding interest and several attempts \cite{thomas-logical-1984,  potthoff-first-order-1995}. 
For a range of tree and forest logics, decidable characterizations have been obtained (see \cite{bojanczyk-effective-2008} for a survey up to 2008; more recent results include \cite{bojanczyk-piecewise-2012,bojanczyk-tree-2010,place-deciding-2010,place-decidable-2009,benedikt-regular-2009} among others).
However, for many more prominent logics, including First-Order Logic with ancestor, $CTL$, $CTL^*$, PDL, and Chain Logic, this problem remains open.
As described in the introduction, PDL extends both $CTL$ and $CTL^*$.
All of these logics were shown by \cite{bojanczyk-wreath-2012} to correspond to iterated wreath products of specific types of forest algebras satisfying a distributive law.
Among these, PDL stands out because it is characterized through products of arbitrary distributive forest algebras, and can -- at least at the level $k = 2$ -- be captured in terms of a $k$-fold iterated distributive law.
Therefore, our results might also shed light on this larger family of open problems.

In the field of regular word languages and logic on words, the study of finite monoids has been tremendously successful.
Our proof strategy highlights how the classical theory developed for studying logic on words via wreath products of monoids carries over faithfully to the setting of forest algebras:
Our proof proceeds by solving two sub-problems related to the left and right factors in wreath product decompositions: a separation result and a Local-Global theorem, which are then combined via the Derived Category Theorem \cite{tilson-categories-1987}.
These steps are remarkably similar to results from the theory of logic on words and finite monoids which also reduce the problem of decidability to separation \cite{place-going-2014} and Local-Global theorems \cite{krebs-effective-2012}.

\bibliography{references}


\begin{thebibliography}{21}


\ifx \showCODEN    \undefined \def \showCODEN     #1{\unskip}     \fi
\ifx \showDOI      \undefined \def \showDOI       #1{#1}\fi
\ifx \showISBNx    \undefined \def \showISBNx     #1{\unskip}     \fi
\ifx \showISBNxiii \undefined \def \showISBNxiii  #1{\unskip}     \fi
\ifx \showISSN     \undefined \def \showISSN      #1{\unskip}     \fi
\ifx \showLCCN     \undefined \def \showLCCN      #1{\unskip}     \fi
\ifx \shownote     \undefined \def \shownote      #1{#1}          \fi
\ifx \showarticletitle \undefined \def \showarticletitle #1{#1}   \fi
\ifx \showURL      \undefined \def \showURL       {\relax}        \fi
\providecommand\bibfield[2]{#2}
\providecommand\bibinfo[2]{#2}
\providecommand\natexlab[1]{#1}
\providecommand\showeprint[2][]{arXiv:#2}

\bibitem[\protect\citeauthoryear{Almeida}{Almeida}{1996}]%
        {almeida-syntactical-1996}
\bibfield{author}{\bibinfo{person}{Jorge Almeida}.}
  \bibinfo{year}{1996}\natexlab{}.
\newblock \showarticletitle{A syntactical proof of locality of {DA}}.
\newblock \bibinfo{journal}{\emph{International Journal of Algebra and
  Computation}} \bibinfo{volume}{06}, \bibinfo{number}{02}
  (\bibinfo{date}{April} \bibinfo{year}{1996}), \bibinfo{pages}{165--177}.
\newblock


\bibitem[\protect\citeauthoryear{Benedikt and Segoufin}{Benedikt and
  Segoufin}{2009}]%
        {benedikt-regular-2009}
\bibfield{author}{\bibinfo{person}{Michael Benedikt} {and} \bibinfo{person}{Luc
  Segoufin}.} \bibinfo{year}{2009}\natexlab{}.
\newblock \showarticletitle{Regular tree languages definable in {{FO}} and in
  {FO}\({}\_{\mbox{\emph{mod}}}\)}.
\newblock \bibinfo{journal}{\emph{{ACM} Trans. Comput. Log.}}
  \bibinfo{volume}{11}, \bibinfo{number}{1} (\bibinfo{date}{Oct.}
  \bibinfo{year}{2009}), \bibinfo{pages}{1--32}.
\newblock


\bibitem[\protect\citeauthoryear{Boja{\'n}czyk}{Boja{\'n}czyk}{2008}]%
        {bojanczyk-effective-2008}
\bibfield{author}{\bibinfo{person}{Mikolaj Boja{\'n}czyk}.}
  \bibinfo{year}{2008}\natexlab{}.
\newblock \showarticletitle{Effective characterizations of tree logics}. In
  \bibinfo{booktitle}{\emph{Proceedings of the {Twenty}-{Seventh} {{ACM}}
  {{SIGMOD}-{SIGACT}-{SIGART}} {Symposium} on {Principles} of {Database}
  {Systems}, {{PODS}} 2008, {June} 9-11, 2008, {Vancouver}, {BC}, {Canada}}},
  \bibfield{editor}{\bibinfo{person}{Maurizio Lenzerini} {and}
  \bibinfo{person}{Domenico Lembo}} (Eds.). \bibinfo{publisher}{ACM},
  \bibinfo{pages}{53--66}.
\newblock
\showISBNx{978-1-60558-108-8}


\bibitem[\protect\citeauthoryear{Bojanczyk and Segoufin}{Bojanczyk and
  Segoufin}{2010}]%
        {bojanczyk-tree-2010}
\bibfield{author}{\bibinfo{person}{Mikolaj Bojanczyk} {and}
  \bibinfo{person}{Luc Segoufin}.} \bibinfo{year}{2010}\natexlab{}.
\newblock \showarticletitle{Tree {Languages} {Defined} in {First}-{Order}
  {Logic} with {One} {Quantifier} {Alternation}}.
\newblock \bibinfo{journal}{\emph{Logical Methods in Computer Science}}
  \bibinfo{volume}{6}, \bibinfo{number}{4} (\bibinfo{year}{2010}).
\newblock


\bibitem[\protect\citeauthoryear{Boja{\'n}czyk, Segoufin, and
  Straubing}{Boja{\'n}czyk et~al\mbox{.}}{2012a}]%
        {bojanczyk-piecewise-2012}
\bibfield{author}{\bibinfo{person}{Mikolaj Boja{\'n}czyk}, \bibinfo{person}{Luc
  Segoufin}, {and} \bibinfo{person}{Howard Straubing}.}
  \bibinfo{year}{2012}\natexlab{a}.
\newblock \showarticletitle{Piecewise testable tree languages}.
\newblock \bibinfo{journal}{\emph{Logical Methods in Computer Science}}
  \bibinfo{volume}{8}, \bibinfo{number}{3} (\bibinfo{year}{2012}),
  \bibinfo{pages}{442--451}.
\newblock


\bibitem[\protect\citeauthoryear{Boja{\'n}czyk, Straubing, and
  Walukiewicz}{Boja{\'n}czyk et~al\mbox{.}}{2012b}]%
        {bojanczyk-wreath-2012}
\bibfield{author}{\bibinfo{person}{Mikolaj Boja{\'n}czyk},
  \bibinfo{person}{Howard Straubing}, {and} \bibinfo{person}{Igor
  Walukiewicz}.} \bibinfo{year}{2012}\natexlab{b}.
\newblock \showarticletitle{Wreath {Products} of {Forest} {Algebras}, with
  {Applications} to {Tree} {Logics}}.
\newblock \bibinfo{journal}{\emph{Logical Methods in Computer Science}}
  \bibinfo{volume}{8}, \bibinfo{number}{3} (\bibinfo{year}{2012}).
\newblock


\bibitem[\protect\citeauthoryear{Brzozowski and Simon}{Brzozowski and
  Simon}{1973}]%
        {brzozowski-characterizations-1973}
\bibfield{author}{\bibinfo{person}{J.~A. Brzozowski} {and}
  \bibinfo{person}{Imre Simon}.} \bibinfo{year}{1973}\natexlab{}.
\newblock \showarticletitle{Characterizations of locally testable events}.
\newblock \bibinfo{journal}{\emph{Discrete Mathematics}} \bibinfo{volume}{4},
  \bibinfo{number}{3} (\bibinfo{date}{March} \bibinfo{year}{1973}),
  \bibinfo{pages}{243--271}.
\newblock


\bibitem[\protect\citeauthoryear{Eilenberg and Tilson}{Eilenberg and
  Tilson}{1976}]%
        {eilenberg-automata-1976}
\bibfield{author}{\bibinfo{person}{S. Eilenberg} {and} \bibinfo{person}{B.
  Tilson}.} \bibinfo{year}{1976}\natexlab{}.
\newblock \bibinfo{booktitle}{\emph{Automata, languages, and machines}}.
\newblock \bibinfo{publisher}{Academic Pr}, \bibinfo{address}{New York}.
\newblock
\showISBNx{978-0-12-234002-4}


\bibitem[\protect\citeauthoryear{Krebs and Straubing}{Krebs and
  Straubing}{2012}]%
        {krebs-effective-2012}
\bibfield{author}{\bibinfo{person}{Andreas Krebs} {and} \bibinfo{person}{Howard
  Straubing}.} \bibinfo{year}{2012}\natexlab{}.
\newblock \showarticletitle{An effective characterization of the alternation
  hierarchy in two-variable logic}. In \bibinfo{booktitle}{\emph{FSTTCS}}.
  \bibinfo{pages}{86--98}.
\newblock


\bibitem[\protect\citeauthoryear{Pin, Straubing, and Th{\'e}rien}{Pin
  et~al\mbox{.}}{1988}]%
        {pin-locally-1988}
\bibfield{author}{\bibinfo{person}{Jean-{\'E}ric Pin}, \bibinfo{person}{Howard
  Straubing}, {and} \bibinfo{person}{Denis Th{\'e}rien}.}
  \bibinfo{year}{1988}\natexlab{}.
\newblock \showarticletitle{Locally trivial categories and unambiguous
  concatenation}.
\newblock \bibinfo{journal}{\emph{Journal of Pure and Applied Algebra}}
  \bibinfo{volume}{52}, \bibinfo{number}{3} (\bibinfo{year}{1988}),
  \bibinfo{pages}{297--311}.
\newblock


\bibitem[\protect\citeauthoryear{Place and Segoufin}{Place and
  Segoufin}{2009}]%
        {place-decidable-2009}
\bibfield{author}{\bibinfo{person}{Thomas Place} {and} \bibinfo{person}{Luc
  Segoufin}.} \bibinfo{year}{2009}\natexlab{}.
\newblock \showarticletitle{A {Decidable} {Characterization} of {Locally}
  {Testable} {Tree} {Languages}}. In \bibinfo{booktitle}{\emph{Automata,
  {Languages} and {Programming}, 36th {Internatilonal} {Collogquium}, {{ICALP}}
  2009, {Rhodes}, greece, {July} 5-12, 2009, {Proceedings}, {Part} {{II}}}}
  \emph{(\bibinfo{series}{Lecture {Notes} in {Computer} {Science}})},
  \bibfield{editor}{\bibinfo{person}{Susanne Albers}, \bibinfo{person}{Alberto
  Marchetti-Spaccamela}, \bibinfo{person}{Yossi Matias},
  \bibinfo{person}{Sotiris~E Nikoletseas}, {and} \bibinfo{person}{Wolfgang
  Thomas}} (Eds.), Vol.~\bibinfo{volume}{5556}. \bibinfo{publisher}{Springer},
  \bibinfo{pages}{285--296}.
\newblock
\showISBNx{978-3-642-02929-5}


\bibitem[\protect\citeauthoryear{Place and Segoufin}{Place and
  Segoufin}{2010}]%
        {place-deciding-2010}
\bibfield{author}{\bibinfo{person}{Thomas Place} {and} \bibinfo{person}{Luc
  Segoufin}.} \bibinfo{year}{2010}\natexlab{}.
\newblock \showarticletitle{Deciding {Definability} in
  {FO}{$_2$}({\textless}{$_h$}, {\textless}{$_v$}) on {Trees}}. In
  \bibinfo{booktitle}{\emph{2010 25th {Annual} {IEEE} {Symposium} on {Logic} in
  {Computer} {Science}}}. \bibinfo{publisher}{IEEE}, \bibinfo{pages}{253--262}.
\newblock
\showISBNx{978-1-4244-7588-9}


\bibitem[\protect\citeauthoryear{Place and Zeitoun}{Place and Zeitoun}{2014}]%
        {place-going-2014}
\bibfield{author}{\bibinfo{person}{Thomas Place} {and} \bibinfo{person}{Marc
  Zeitoun}.} \bibinfo{year}{2014}\natexlab{}.
\newblock \showarticletitle{Going {Higher} in the {First}-{Order} {Quantifier}
  {Alternation} {Hierarchy} on {Words}}. In \bibinfo{booktitle}{\emph{Automata,
  {Languages}, and {Programming} - 41st {International} {Colloquium}, {{ICALP}}
  2014, {Copenhagen}, {Denmark}, {July} 8-11, 2014, {Proceedings}, {Part}
  {{II}}}} \emph{(\bibinfo{series}{Lecture {Notes} in {Computer} {Science}})},
  \bibfield{editor}{\bibinfo{person}{Javier Esparza}, \bibinfo{person}{Pierre
  Fraigniaud}, \bibinfo{person}{Thore Husfeldt}, {and} \bibinfo{person}{Elias
  Koutsoupias}} (Eds.), Vol.~\bibinfo{volume}{8573}.
  \bibinfo{publisher}{Springer}, \bibinfo{pages}{342--353}.
\newblock
\showISBNx{978-3-662-43950-0}


\bibitem[\protect\citeauthoryear{Potthoff}{Potthoff}{1995}]%
        {potthoff-first-order-1995}
\bibfield{author}{\bibinfo{person}{Andreas Potthoff}.}
  \bibinfo{year}{1995}\natexlab{}.
\newblock \showarticletitle{First-{Order} {Logic} on {Finite} {Trees}}. In
  \bibinfo{booktitle}{\emph{{TAPSOFT}'95: {Theory} and {Practice} of {Software}
  {Development}, 6th {International} {Joint} {Conference} {CAAP}/{FASE},
  {Aarhus}, {Denmark}, {May} 22-26, 1995, {Proceedings}}}
  \emph{(\bibinfo{series}{Lecture {Notes} in {Computer} {Science}})},
  \bibfield{editor}{\bibinfo{person}{Peter~D Mosses}, \bibinfo{person}{Mogens
  Nielsen}, {and} \bibinfo{person}{Michael~I Schwartzbach}} (Eds.),
  Vol.~\bibinfo{volume}{915}. \bibinfo{publisher}{Springer},
  \bibinfo{pages}{125--139}.
\newblock
\showISBNx{3-540-59293-8}


\bibitem[\protect\citeauthoryear{Rhodes and Steinberg}{Rhodes and
  Steinberg}{2009}]%
        {rhodes-q-theory-2009}
\bibfield{author}{\bibinfo{person}{John Rhodes} {and} \bibinfo{person}{Benjamin
  Steinberg}.} \bibinfo{year}{2009}\natexlab{}.
\newblock \bibinfo{booktitle}{\emph{The q-theory of finite semigroups}}.
\newblock \bibinfo{publisher}{Springer Science \& Business Media}.
\newblock


\bibitem[\protect\citeauthoryear{Straubing}{Straubing}{1985}]%
        {straubing-finite-1985}
\bibfield{author}{\bibinfo{person}{Howard Straubing}.}
  \bibinfo{year}{1985}\natexlab{}.
\newblock \showarticletitle{Finite semigroup varieties of the form {V} * {D}}.
\newblock \bibinfo{journal}{\emph{Journal of Pure and Applied Algebra}}
  \bibinfo{volume}{36} (\bibinfo{year}{1985}), \bibinfo{pages}{53--94}.
\newblock


\bibitem[\protect\citeauthoryear{Straubing}{Straubing}{2013}]%
        {straubing-new-2013}
\bibfield{author}{\bibinfo{person}{Howard Straubing}.}
  \bibinfo{year}{2013}\natexlab{}.
\newblock \showarticletitle{New applications of the wreath product of forest
  algebras}.
\newblock \bibinfo{journal}{\emph{{RAIRO} - Theor. Inf. and Applic.}}
  \bibinfo{volume}{47}, \bibinfo{number}{3} (\bibinfo{year}{2013}),
  \bibinfo{pages}{261--291}.
\newblock


\bibitem[\protect\citeauthoryear{Straubing}{Straubing}{2018}]%
        {straubing-forest-2018}
\bibfield{author}{\bibinfo{person}{Howard Straubing}.}
  \bibinfo{year}{2018}\natexlab{}.
\newblock \showarticletitle{Forest {Categories}}.
\newblock \bibinfo{journal}{\emph{arXiv:1801.04337 [cs]}} (\bibinfo{date}{Jan.}
  \bibinfo{year}{2018}).
\newblock
\newblock
\shownote{arXiv: 1801.04337.}


\bibitem[\protect\citeauthoryear{Th{\'e}rien and Weiss}{Th{\'e}rien and
  Weiss}{1985}]%
        {therien-graph-1985}
\bibfield{author}{\bibinfo{person}{Denis Th{\'e}rien} {and}
  \bibinfo{person}{Alex Weiss}.} \bibinfo{year}{1985}\natexlab{}.
\newblock \showarticletitle{Graph congruences and wreath products}.
\newblock \bibinfo{journal}{\emph{Journal of Pure and Applied Algebra}}
  \bibinfo{volume}{36} (\bibinfo{date}{Jan.} \bibinfo{year}{1985}),
  \bibinfo{pages}{205--215}.
\newblock


\bibitem[\protect\citeauthoryear{Thomas}{Thomas}{1984}]%
        {thomas-logical-1984}
\bibfield{author}{\bibinfo{person}{W. Thomas}.}
  \bibinfo{year}{1984}\natexlab{}.
\newblock \showarticletitle{Logical aspects in the study of tree languages}. In
  \bibinfo{booktitle}{\emph{Ninth {Colloquium} on {Trees} in {Algebra} and in
  {Programming} ({Proc}. {CAAP} 84)}},
  \bibfield{editor}{\bibinfo{person}{{Courcelle}} {and} \bibinfo{person}{{B.}}}
  (Eds.). \bibinfo{publisher}{Cambridge University Press},
  \bibinfo{pages}{31--51}.
\newblock


\bibitem[\protect\citeauthoryear{Tilson}{Tilson}{1987}]%
        {tilson-categories-1987}
\bibfield{author}{\bibinfo{person}{Bret Tilson}.}
  \bibinfo{year}{1987}\natexlab{}.
\newblock \showarticletitle{Categories as algebra: {An} essential ingredient in
  the theory of monoids}.
\newblock \bibinfo{journal}{\emph{Journal of Pure and Applied Algebra}}
  \bibinfo{volume}{48}, \bibinfo{number}{1-2} (\bibinfo{date}{Sept.}
  \bibinfo{year}{1987}), \bibinfo{pages}{83--198}.
\newblock


\end{thebibliography}

\newpage

\appendix*
\section{General Definitions and Lemmas}

\begin{defin}\label{def:pathsets}[Pathsets: $\pi$, $\freedisth$, $\Pi$]
Recall the map $\pi : H_\Sigma \rightarrow Pow(\Sigma^*)$ mapping forests to their path sets.
The image of $\pi$ is called $\freedisth$. That is, $\freedisth$ is the set of nonempty finite subsets of $\Sigma^*$ that are closed under prefixes ($wv \in X \Rightarrow w \in X$).

We extend $\pi$ to an operation on forest languages:
\[\pi(\La) := \{\pi(f) : f \in \La\} \subset Pow(\Sigma^*)\] for $\La \subseteq H_\Sigma$.

When $\La \subseteq H_\Sigma$, then $\Pi(\La)\subseteq \Sigma^*$ is the set of all paths occurring in some forest in $\La$:
$$\Pi(\La) := \bigcup_{f \in \La} \pi(f)$$

Since $\Pi$ is defined in terms of $\pi$, $\Pi$ is also well-defined on subsets of $\freedisth$: If $X \subseteq \freedisth$, then $\Pi(X) := \bigcup_{\phi \in X} \phi$.
\end{defin}

\begin{prop}\label{prop:pi-reg}
Let $\La \subset H_\Sigma$ be a regular forest language. Then $\Pi(\La) \subset \Sigma^*$ is a regular word language, and can be effectively constructed from (an automaton for) $\La$.
\end{prop}

\begin{proof}
Let $\Ff = (H,V)$ be a finite forest algebra recognizing $\La$ via morphism $\phi$. 
We want to construct a finite automaton recognizing the word language $\Pi(\La)$.
Let $Q := Pow(V)$ be the state set of this automaton, and consider the transition function $Q \times \Sigma \rightarrow Q$
given by
\[ (\{v_1, ..., v_k\} , \alpha) \mapsto \{v : \exists j =1, ...,k, \exists h \in H : v = v_j \cdot \phi(\alpha) \cdot I_h \}\]
where by $\phi(\alpha)$ we refer to the image of the context consisting of only $\alpha$ and a variable below it ($\alpha[X]$), and $I_h \in V$ is the insertion operation defined in Definition~\ref{def:forest-algebra}.
We take the starting state to be $\{1_V\} \in Q$, and the accepting states to be all nonempty subsets of $\{v : \phi^{-1}(v \cdot 0_H) \subset \La \}$.
By induction over the length of words, it is shown that this automaton exactly accepts the language $\Pi(\La)$.
\end{proof}

For reference, we formulate the full version of the Derived Category Theorem.
It is not required for our main result, but will be used in Appendix~\ref{sec:proof-ex} for the proof of Example~\ref{ex:langs}.

\begin{theorem}\label{thm:full:deriv}[Full Version of the Derived Category Theorem, \cite{straubing-forest-2018}]
Let $\Sigma$ be an alphabet, and let $\alpha, \beta$ be morphisms from $\Sigma^\Delta$ onto forest algebras $(H_1, V_1)$, $(H_2, V_2)$, respectively. Let $(H,V)$ be a forest algebra.

(a) Assume $D_{\alpha, \beta} \prec (H,V)$. Then
$$(H_1,V_1) \prec (H,V) \wr (H_2, V_2)$$

(b) Suppose $\alpha$ factors as
$$\alpha = \gamma\delta : \Sigma^\Delta \rightarrow (H_1, V_1)$$
where
$$\delta : \Sigma^\Delta \rightarrow (H,V) \wr (H_2,V_2)$$
$$\gamma : \text{Im}\ \delta \rightarrow (H_1,V_1)$$
and that $\beta = \pi^{(2)}\delta$, where $\pi^{(2)}$ is the projection homomorphism from the wreath product onto its right-hand factor. Then $D_{\alpha,\beta} \prec (H,V)$.
\end{theorem}

\newpage
\section{Proving the Local-Global Theorem}\label{sec:locglob}

\begin{figure}
\begin{subfigure}[b]{0.3\textwidth}
\begin{center}
\begin{tikzpicture}[level/.style={sibling distance=40mm/#1}, scale=0.7,transform shape]
\node [circle,draw] (a) {$h \xleftarrow{v} \left(h' + h''\right)$}
    child {node [circle,draw] (b) {$h' \xleftarrow{h_1}$}
    }
    child {node [circle,draw] (c) {$h'' \xleftarrow{h_2}$}
    };
\end{tikzpicture}
\end{center}
\caption{}\label{fig:diag-basic}
\end{subfigure}
\begin{subfigure}[b]{0.3\textwidth}
\begin{center}
\begin{tikzpicture}[level/.style={sibling distance=40mm/#1}, scale=0.7,transform shape]
\begin{scope}[xshift=-1cm]
\node [circle,draw] (a) {$h' \xleftarrow{v} h$}
    child {node [circle,draw] (b) {$h \xleftarrow{h_1}$}
    };
\end{scope}
\begin{scope}[xshift=1cm]
\node [circle,draw] (c) {$h' \xleftarrow{v} h$}
    child {node [circle,draw] (d) {$h \xleftarrow{h_2}$}
    };
\end{scope}
\end{tikzpicture}
\end{center}
\caption{}\label{fig:diag-sep}
\end{subfigure}
\begin{subfigure}[b]{0.3\textwidth}
\begin{center}
\begin{tikzpicture}[level/.style={sibling distance=40mm/#1}, scale=0.7,transform shape]
\node [circle,draw] (a) {$h' \xleftarrow{v} h$}
    child {node [circle,draw] (b) {$h \xleftarrow{h_1}$}
    }
    child {node [circle,draw] (c) {$h \xleftarrow{h_2}$} };
\end{tikzpicture}
\end{center}
\caption{}\label{fig:diag-together}
\end{subfigure}
\caption{Examples of forest diagrams. In any locally distributive derived forest category, the diagrams in (b) and (c) result in the same value under the map $\operatorname{Val}$.}\label{fig:diagrams}
\end{figure}
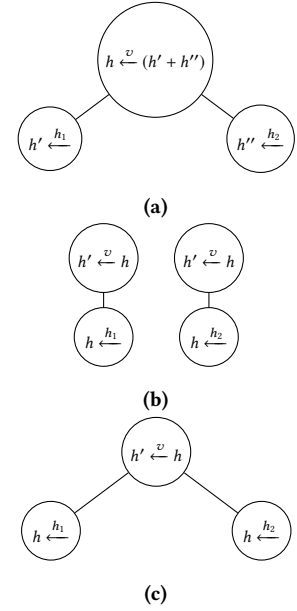

We will need the notion of \emph{forest diagrams}:

\begin{defin}
Given a derived forest category $C$, a \emph{forest diagram} is a forest whose leaves are labeled with half-arrows of $C$ and whose internal nodes are labeled with arrows of $C$, subject to the following consistency condition: Let $n$ be a (non-leaf) node labeled with an arrow $\carrow{h_1}{v}{h_2}$, and let $h_3, ..., h_k \in \Obj(C)$ be the endpoints of the arrows or half-arrows that label the root nodes of the children of $n$. Then $h_3 + ... + h_k = h_1$.
\end{defin}

Examples are shown in Figure~\ref{fig:diagrams}.

To any forest diagram one can assign a half-arrow of $C$ by recursively adding up the half-arrows assigned to siblings, and multiplying out the action of arrows on the half-arrows assigned to their children.
Let $\operatorname{Val}$ be the map assigning these values to forest-diagrams.

One could think of forest diagrams as forming a `free forest category' -- in analogy to $\Sigma^\Delta$, and $\operatorname{Val}$ being a `forest category morphism' to $C$, but we will not develop this perspective here.

Assume that $C$ is a finite and locally-distributive derived forest category.
Let $d$ be a forest diagram over $C$.
Let us assume that $d$ contains two sibling nodes labeled with the same arrow $\carrow{h}{v}{h'}$ (see Figure~\ref{fig:diag-sep}).
Let $d_1, d_2$ be the forest diagrams below these two nodes.
The value of the sum of these two sibling nodes is
$$\left(\carrow{h}{v}{h'}\right) \cdot \operatorname{Val}(d_1) + \left(\carrow{h}{v}{h'}\right) \cdot \operatorname{Val}(d_2)$$
Since $C$ is locally distributive, we know that this is equal to
$$\left(\carrow{h}{v}{h'}\right) \cdot \operatorname{Val}(d_1 + d_2)$$
Therefore, if we replace the two sibling nodes by a single node and put $d_1+d_2$ below it (see Figure~\ref{fig:diag-together}), we obtain a new forest diagram which evaluates to the same half-arrow in $C$ as $d$ did.

If we apply this procedure iteratively, we will, after finitely many steps, arrive at a forest diagram $\widehat{d}$ where no two sibling nodes are labeled by the same half-arrow, and for which $\operatorname{Val}(\widehat{d}) = \operatorname{Val}(d)$.
Indeed, taking $\Psi$ from Definition~\ref{def:psi}, the result $\widehat{d}$ is equal to $\Psi(d)$.

Let $\Sigma_C$ be the alphabet consisting of all arrows and half-arrows of $C$.
Then each forest-diagram is a forest in $H_{\Sigma_C}$.

With this view, $\widehat{d}$ and $d$ have the same pathsets: $\pi(\widehat{d}) = \pi(d)$ (Proposition~\ref{prop:psi}.2).
Furthermore, it is not hard to see that for any other forest diagram $d'$ with $\pi(d') = \pi(d)$, we would have arrived at the same modified diagram $\widehat{d}$  (Proposition~\ref{prop:psi}.3).
This implies that $\operatorname{Val}(d)$ is determined by $\pi(d)$. 

Having established this, we can prove the theorem:
\begin{proof}[Proof of the Local-Global Theorem~\ref{thm:loc-glob}]
Assume that $C$ is a finite and locally-distributive derived forest category.

Let $\carrow{}{c}h$, $\carrow{}{c'}{h'}$ be distinct half-arrows in $C$.
By the previous argument, $$\pi(\operatorname{Val}^{-1}(\carrow{}{c}h)) \cap \pi(\operatorname{Val}^{-1}(\carrow{}{c'}{h'})) = \emptyset$$

Certainly, the sum of the two half-arrows has to be distinct from at least one of them:
At least, $\carrow{}{c}h + \carrow{}{c'}{h'} \neq \carrow{}{c}h$ or $\carrow{}{c}h + \carrow{}{c'}{h'} \neq \carrow{}{c'}{h'}$.
We'll assume that the first one holds.
Then, we know
$\pi(\operatorname{Val}^{-1}(\carrow{}{c}h + \carrow{}{c'}{h'})) \cap \pi(\operatorname{Val}^{-1}(\carrow{}{c}h)) = \emptyset$

Recall the operation $\Pi$ from Definition~\ref{def:pathsets}:
For any forest language $\La$, we define $\Pi(\La)\subseteq (\Sigma_C)^*$ as the set of all paths occurring in some forest in $\La$:
$$\Pi(\La) := \bigcup_{f \in \La} \pi(f)$$

Now consider the sets
$\Pi_1 := \Pi(\operatorname{Val}^{-1}(\carrow{}{c}h))$ and 
$\Pi_2 := \Pi(\operatorname{Val}^{-1}(\carrow{}{c'}{h'}))$.
Assume that there is $d_1 \in \operatorname{Val}^{-1}(\carrow{}{c}h)$ such that $\pi(d_1) \subseteq \Pi_2$.
Then we can find $d_3 \in \operatorname{Val}^{-1}(\carrow{}{c'}{h'})$ such that $\pi(d_1) \subseteq \pi(d_3)$.
Observe $\pi(d_1) = \pi(d_3 + d_1)$ and $d_3 + d_1 \in \operatorname{Val}^{-1}(\carrow{}{c}h + \carrow{}{c'}{h'})$.
This means $\pi(d_1) \in \pi(\operatorname{Val}^{-1}(\carrow{}{c}h + \carrow{}{c'}{h'})) \cap \pi(\operatorname{Val}^{-1}(\carrow{}{c}h))$, contradiction.

This means that there is no $d_1 \in \operatorname{Val}^{-1}(\carrow{}{c}h)$ such that $\pi(d_1) \subseteq \Pi_2$.
Said differently, for any $d_1 \in \operatorname{Val}^{-1}(\carrow{}{c}h)$, we have $\pi(d_1) \cap (\Pi_1-\Pi_2) \neq \emptyset$.

This means that, in order to separate forest diagrams evaluating to $\carrow{}{c}h$ from those evaluating to $\carrow{}{c'}{h'}$, it is sufficient to check whether the pathset of the diagram in question contains an element of $\Pi_1-\Pi_2$.

We want to show that $\Pi_1, \Pi_2 \subset \Sigma_C^*$ are regular word languages over the finite alphabet $\Sigma_C$.
Being a finite derived forest category, $C$ is equal to $D_{\phi ,\psi}$ for some forest  algebra morphisms $\phi, \psi$ into finite forest algebras.
Thus, $\operatorname{Val}^{-1}(\carrow{}{c}h)$ and $\operatorname{Val}^{-1}(\carrow{}{c'}{h'})$ are recognized by combining a morphism into the (finite) image of $\phi$ that disregards the start- and end-points of (half-)arrows and otherwise behaves like $\phi$, with a regular language that checks whether the start- and end-points of the node labels are locally consistent.
By Lemma~\ref{prop:pi-reg}, it follows that $\Pi_1$ and $\Pi_2$ are regular word languages over $\Sigma_C$.

Thus, by Proposition~\ref{prop:distr-char}, there is a finite distributive forest algebra $\Ff$ and a morphism $\eta : (\Sigma_C)^\Delta \rightarrow \Ff$ which recognizes the set of forests whose path set intersects $\Pi_1-\Pi_2$.
While we cannot hope that this morphism recognizes the set of such forest diagrams, it certainly separates the class of diagrams whose path sets intersect $\Pi_1-\Pi_2$ from those for which this doesn't hold.

For each pair of half-arrows $\carrow{}{c}h$, $\carrow{}{c'}{h'}$ in $C$, we obtain a finite distributive forest algebra $\Ff_{\carrow{}{c}h, \carrow{}{c'}{h'}}$ and a morphism $\eta_{\carrow{}{c}h, \carrow{}{c'}{h'}}$ in this manner.
We then let $\widehat{\Ff}$ be the direct product of these finitely many algebras, and $\widehat{\eta}$ the direct product of these morphisms.
Certainly, $\widehat{\Ff}$ is finite and distributive.
To construct a division $C \prec \widehat{\Ff}$ (recall Definition~\ref{def:division}), we assign to each half-arrow 
$\carrow{}{c}h$ in $C$ the set of half-arrows in $\widehat{\Ff}$ that are in the image of $\operatorname{Val}^{-1}(\carrow{}{c}h)$ under $\widehat{\eta}$, and similarly for arrows.
One can verify that these assignments preserve operations as required in Definition~\ref{def:division}.
More importantly, the previous arguments guarantee that they satisfy the injectivity condition required in Definition~\ref{def:division}:

For a contradiction, let $\carrow{}{c}h$, $\carrow{}{c'}{h'}$ be distinct half-arrows in $C$ that are mapped to overlapping sets in $H_{\widehat{\Ff}}$:
$g \in K_{\carrow{}{c}h} \cap K_{\carrow{}{c'}{h'}}$
This means that there are forest diagrams $f_1, f_2$ such that
$\operatorname{Val}(f_1) = K_{\carrow{}{c}h}$, $\operatorname{Val}(f_2) = K_{\carrow{}{c'}{h'}}$, but $\widehat{\eta}(f_1) = \widehat{\eta}(f_2) = g$.
In particular, $\eta_{\carrow{}{c}h, \carrow{}{c'}{h'}}(f_1) = \eta_{\carrow{}{c}h, \carrow{}{c'}{h'}}(f_2)$.
By the choice of $\eta_{\carrow{}{c}h, \carrow{}{c'}{h'}}$, this is impossible.
The argument for arrows is similar.

In conclusion, we have shown that the assignment is injective and $C$ divides the finite distributive forest algebra $\widehat{\Ff}$.

\end{proof}

\newpage
\section{Proving the Separation Lemma}\label{sec:sep}

Recall from the beginning that trees and forests are horizontally commutative and idempotent.
Informally, order and multiplicity of children in a tree and of trees in a forest don't matter.
To fix notation, we denote a tree with root symbol $\alpha$ and children set $C = \{t_1, ..., t_n\}$ as $\alpha[C]$.

We are given regular forest languages $\La_1, \La_2$. We want to show that, if the images of $\La_1, \La_2$ under $\pi$ are disjoint, then some language recognized by a wreath product of three finite distributive algebras separates $\La_1, \La_2$.

We will prove a stronger statement to `load the induction hypothesis'.
This stronger statement will be the Main Lemma~\ref{prop:main-lemma}.
In order to state and prove it, a few more notions are needed.

\subsection{Preliminaries}

We will need a notion of \emph{rules} that is analogous to transitions in bottom-up tree automata:

\begin{defin}[Rules]\label{def:algebras}\label{d:p:ls-lr}
Let $\Ff = (H,V)$ be a forest algebra, and $\phi : \Sigma^\Delta \rightarrow (H,V)$ a morphism.

A \emph{rule} is a tuple $$(h_0, \alpha, \{h_1, ..., h_k\}) \in H \times \Sigma \times Pow(H)$$ such that $$h_0 = \phi(\alpha)[h_1 + ... + h_k]$$
The set of rules for $\Ff = (H,V)$ with morphism $\phi$ is called $Rules(\Ff, \phi)$.

 For $\rho \in Rules(\Ff)$, let $\La(\rho)$ be the set of nonempty forests where all top-level rule applications are of rule $\rho = (h_0, \alpha, \{h_1, ..., h_k\})$.
Formally, a nonempty forest $f$ belongs to $\La(\rho)$ iff, for any tree $t \in f$, (1) the root symbol of $t$ is $\alpha$, and (2) if $C$ is the set of trees that are children of the root node of $t$, then $\{\phi(c) :c \in C\} = \{h_1, ... , h_k\}$.

For convenience, we will use the same notation for elements of $H$: $\La(h) := \phi^{-1}(h)$ for $h \in H$.

\end{defin}

\subsection{Main Lemma and Proof of Separation Lemma}

\begin{llemma}[Main Lemma]\label{prop:main-lemma}
Let $\Ff = (H,V)$ be a finite forest algebra, and $\phi : \Sigma^\Delta\rightarrow(H,V)$ a morphism.
There is a map $\Pa$ assigning to each $\rho \in Rules(\Ff)$ a language $\Pa(\rho) \subseteq H_\Sigma$ such that the following four statements hold: 

(A) For any $\rho \in Rules(\Ff,\phi)$, we have $$\La(\rho) \subseteq \Pa(\rho)$$

(B) For all $\rho, \rho' \in Rules(\Ff,\phi)$,
$\Pa(\rho) \cap \Pa(\rho')$ is empty if and only if $\pi\La(\rho) \cap \pi\La(\rho')$ is empty.

(C) For any $\rho_1, ... \rho_k \in Rules(\Ff,\phi)$, the language $$\{f = f_1 + ... + f_k : f_i \in \Pa(\rho_i)\}$$
is recognized by a wreath product of three finite distributive forest algebras.

(D) For any $\rho \in Rules(\Ff,\phi)$ and any forest $f = \{t_1, ..., t_l\}$, where each $t_i$ is a tree, we have $f \in \Pa(\rho)$ if and only if all the  singleton forests $\{t_i\}$ ($i = 1, ..., l$) are in $\Pa(\rho)$.
\end{llemma}

Each $\Pa(\rho)$ can be seen as a PDL approximator of the language $\La(\rho)$.
While (A) guarantees that these approximators are sufficiently `large', (B) guarantees that they suffice to construct a PDL separator whenever $\La(\rho)$ and $\La(\rho')$ are separable by $\pi$ at all.
We'll prove the Main Lemma below in Appendix~\ref{sec:main-lemma}. Using the Main Lemma, we prove the Separation Lemma:

\begin{proof}[Proof of the Separation Lemma~\ref{cor:separation}]
Since $\La_1, \La_2$ are regular, there is a finite forest algebra $\Ff = (H,V)$ which recognizes both $\La_1$ and $\La_2$ via a morphism $\phi : \Sigma^\Delta \rightarrow\Ff$.

Any forest $f$ in $H_\Sigma$ can be uniquely written as a sum of forests from languages $\La(\rho)$, with $\rho \in Rules(\Ff,\phi)$.
If, for each forest $f$, we use $R_f$ to denote the collection of these $\rho$'s, we can assign to each $h \in H$ a set
$$Q_h := \{ R_f : f \in \phi^{-1}(h)\}$$
We can assume that $\phi$ is onto, so this set is always nonempty.
Each element of this set is a subset of the finite set $Rules(\Ff,\phi)$.
Thus, $Q_h$ itself is finite.
Given the choice of $Q_h$, we can write
$$\phi^{-1}(h) := \bigcup_{q \in Q_h} \sum_{\rho\in q} \La(\rho)$$
Using the map $\Pa$ from the Main Lemma, we then define $$\Pa_h := \bigcup_{q \in Q_h} \sum_{\rho\in q} \Pa(\rho)$$
By condition (C) of the Main Lemma, each $\sum_{\rho\in q} \Pa(\rho)$ is recognized by a wreath product of three finite distributive algebras.
Since $Q_h$ is finite, $\Pa_h$ is also recognized by a wreath product of three finite distributive algebras.

Furthermore, let $\Gg$ be a wreath product of three finite distributive algebras that recognizes \emph{all} languages $\Pa_h$ via a morphism $\phi_\Gg$.
Again, this is possible since $H$ is finite.

Let $h_1, h_2 \in H$.
Assume $\pi(\phi^{-1}(h_1)) \cap \pi(\phi^{-1}(h_2)) = \emptyset$.
Then for all $q \in Q_{h_1}, q' \in Q_{h_2}$, we have
$$\pi\left(\sum_{\rho \in q} \La(\rho)\right) \cap \pi\left(\sum_{\rho \in q'} \La(\rho)\right) = \emptyset$$
Hence, there is $\rho_0 \in q$ such that $\pi(\La(\rho_0)) \cap \pi(\La(\rho')) = \emptyset$ for all $\rho' \in q'$, or the same with $q, q'$ reversed.
\begin{leftbar}
\begin{proof}[Proof of this]
Assume for each $\rho \in q$, there is $\rho' \in q'$ such that $\pi(\La(\rho)) \cap \pi(\La(\rho')) \neq \emptyset$, and the same with $q,q'$ reversed.
First, for $\rho \in q$, let $\phi_\rho \in \pi(\La(\rho)) \cap \pi(\La(\rho'))$.
Then, for $\rho \in q'$, let $\phi'_\rho \in \pi(\La(\rho)) \cap \pi(\La(\rho'))$.
Then define $\widehat\phi := \sum_{\rho\in q} \phi_s + \sum_{\rho \in q'} \phi'_\rho$.
By definition,
$\widehat\phi \in \pi\left(\sum_{\rho \in q} \La(\rho)\right) \cap \pi\left(\sum_{\rho \in q'} \La(\rho)\right)$.
\end{proof}
\end{leftbar}
From condition (B) of the Main Lemma, we can deduce $\Pa(\rho_0) \cap \Pa(\rho') = \emptyset$ for all $\rho' \in q'$ (or $q$, depending on where $\rho_0$ came from). 

From this, we want to deduce
$$\left(\sum_{\rho \in q} \Pa_\rho\right) \cap \left(\sum_{\rho \in q'} \Pa_\rho\right) = \emptyset$$
This follows from condition (D) of the Main Lemma: If this intersection were nonempty and contained a forest $f$, this forest would contain a tree $t$ such that $\{t\} \in \Pa(\rho_0)$ and $\{t\} \in \Pa(\rho')$ for some $\rho' \in q'$ -- a contradiction.

Since $q, q'$ were chosen arbitrarily,
$$\left(\bigcup_{q \in Q_{h_1}} \sum_{\rho \in q} \Pa_\rho\right) \cap \left(\bigcup_{q \in Q_{h_2}} \sum_{\rho \in q} \Pa_\rho\right) = \emptyset$$
which means that $\Pa_{h_1} \cap \Pa_{h_2} = \emptyset$.

Now consider
$$\Pa_1 := \bigcup_{h \in \phi(\La_1)} \Pa_h$$
$$\Pa_2 := \bigcup_{h \in \phi(\La_2)} \Pa_h$$
These languages are recognized by $\Gg$ with morphism $\phi_\Gg$.
Now given $\pi\La_1 \cap \pi\La_2 = \emptyset$, we have $\pi\phi^{-1}(h) \cap \pi\phi^{-1}(h')$ for all $h \in \phi(\La_1)$, $h' \in \phi(\La_2)$.
By our previous reasoning, we get $\Pa_h \cap \Pa_{h'} = \emptyset$ for any such $h, h'$.
Putting this together, we get $\Pa_1 \cap \Pa_2 = \emptyset$, and thus $\Pa_1 \subseteq \left(H_\Sigma - \Pa_2\right)$.

Since $\Ff$ recognizes $\La_1, \La_2$ via the morphism $\phi$, we can write
$$\La_1 = \bigcup_{h \in \phi(\La_1)} \phi^{-1}(h)$$
$$\La_2 = \bigcup_{h \in \phi(\La_2)} \phi^{-1}(h)$$
By condition (A) of the Main Lemma, $\phi^{-1}(h) \subseteq \Pa_h$ holds for any $h \in H$.
Given the definition of $\Pa_1, \Pa_2$, we can concude $\La_1 \subseteq \Pa_1$ and $\La_2 \subseteq \Pa_2$.

Putting these results together, we find that \[\La_1 \subseteq \Pa_1 \subseteq \left(H_\Sigma - \Pa_2\right) \subseteq \left(H_\Sigma - \La_2\right)\]
We can thus take $X := \Pa_1$, completing the proof of the Separation Lemma~\ref{cor:separation}.
\end{proof}

\newpage
\section{Proving the Main Lemma}\label{sec:main-lemma}

The goal in this section is to prove the Main Lemma~\ref{prop:main-lemma}.

For the proof, we will utilize the following bit of notation.
Recall the set $\freedisth$ from Definition~\ref{def:pathsets}.
\begin{defin}
Let $\Ff = (H,V)$.
Let $R, S \subseteq Rules(\Ff,\phi)$ or $R, S \subseteq H$.
Then $\left\langle R, S\right\rangle$ is a subset of $\freedisth$ defined as follows:

$d \in \left\langle R, S\right\rangle$ if and only if there are forests $f_1, f_2$ such that $\pi(f_1) = \pi(f_2) = d$, and there are 
$r_1, ..., r_n \in R$, $s_1, ..., s_k \in S$ ($n, k > 0$) and
$$f_1 \in \La(r_1) + ... + \La(r_n)$$
$$f_2 \in \La(s_1) + ... + \La(s_k)$$
\end{defin}
Recall the operation $\Pi$ mapping forest languages to subsets of $\Sigma^*$.
We will also apply it to subsets of $\freedisth$: For $X \subseteq \freedisth$, set $\Pi(X) := \bigcup_{x \in X} x$.

\begin{defin}\label{def:height}
The \emph{height} of a forest is the length of the longest paths. Formally, we define $\height(\alpha[C]) := 1 + \max_{t \in C} \height(t)$, with $\max \emptyset := 0$.
If $f$ is a forest, we set $\height(f) := \max_{t \in f} \height(t)$, with $\max\emptyset := 0$.

We will frequently use:
$$Z_n(\La) := \{f \in \La: \height f \leq n\}$$

Since the height of a forest only depends on its pathset and the elements of $\freedisth$ are finite sets, we can write $Z_n X$ even when $X \subseteq \freedisth$.

\end{defin}

The goal is to prove the Main Lemma \ref{prop:main-lemma}. We prove the following stronger version.
Recall the operation $\Pi$ from Definition~\ref{def:pathsets}, and $Z_n\La$ from Definition~\ref{def:height}.

\begin{llemma}[Stronger Version of Main Lemma]\label{prop:main-lemma-strong}
Let $\Ff = (H,V)$ be a finite forest algebra, and $\phi : \Sigma^\Delta\rightarrow(H,V)$ a morphism.
There is a map $\Pa$ assigning to each $R \subseteq Rules(\Ff)$ a language $\Pa(R) \subseteq H_\Sigma$ such that the following four statements hold: 

(A) For any $R \subseteq Rules(\Ff,\phi)$, and for any $\rho_1, ..., \rho_k \in R$ ($r \geq 1$), we have $$\La(\rho_1) + ... + \La(\rho_k) \subseteq \Pa(R)$$

That is, $\Pa(R)$ extends the languages obtained by adding forests that evaluate to rules in $R$.

(B) The following relationship between $\Pa(\cdot)$ and the construct $\left\langle\cdot,\cdot\right\rangle$ holds
for all $R, S \subseteq Rules(\Ff,\phi)$, and $n \in \N$: 
$$\Pi \left(\Rn\Pa(R) \cap \Rn\Pa(S)\right) =  \Pi \Rn \left\langle R, S \right\rangle$$
In words: The following two operations yield the same sets of paths:

\begin{enumerate}
\item Intersect the languages $\Pa(R)$, $\Pa(S)$, restricted to forests of height $\leq n$. Then compute the set of paths occurring in the resulting forest language. 
\item Compute the set $\left\langle R, S\right\rangle \subseteq \freedisth$, and restrict to pathsets of height $\leq n$. Then compute the set of paths, that is, the union over the resulting subset of $\freedisth$.
\end{enumerate}

(C) For any $\rho_1, ... \rho_k \in Rules(\Ff,\phi)$, the language $$\{f = f_1 + ... + f_k : f_i \in \Pa(\{\rho_i\})\}$$
is recognized by a wreath product of three finite distributive forest algebras.

(D) For any $R \subseteq Rules(\Ff,\phi)$ and any forest $f = \{t_1, ..., t_l\}$ (each $t_i$ being a tree), we have $f \in \Pa(R)$ if and only if all the singleton forests $\{t_i\}$ ($i = 1, ..., l$) are in $\Pa(R)$.
\end{llemma}

To obtain Lemma~\ref{prop:main-lemma}, we take $\Pa(\rho)$ from that version to be the map $\Pa$ given here applied to the singleton $\{\rho\}$. 
Conditions (A), (C), and (D) immediately follow.
For condition (B), note that $\Pa(\{\rho\}) \cap \Pa(\{\rho'\})$ is empty if and only if $\Rn\Pa(\{\rho\}) \cap \Rn\Pa(\{\rho'\})$ is empty for all $n$.
Now, since even empty forests have a nonempty pathset (consisting of the empty path), $\Rn\Pa(\{\rho\}) \cap \Rn\Pa(\{\rho'\})$ is empty if and only if $\Pi(\Rn\Pa(\{\rho\}) \cap \Rn\Pa(\{\rho'\}))$ is.
Now, by (B) from Lemma~\ref{prop:main-lemma-strong}, this is empty if and only if $\Pi \Rn \left\langle\{\rho\}, \{\rho'\}\right\rangle$ is empty.
By definition of $\langle\cdot,\cdot\rangle$, the term $\left\langle\{\rho\}, \{\rho'\}\right\rangle$ is actually equal to $\pi\La(\rho) \cap \pi\La(\rho')$.
Reversing the previous arguments from this paragraph, $\Pi \Rn (\pi\La(\rho) \cap \pi\La(\rho'))$ is empty for all $n$ if and only if $\pi\La(\rho) \cap \pi\La(\rho')$ is empty.
We have proven Condition (B) from Lemma~\ref{prop:main-lemma}.
This shows that Lemma~\ref{prop:main-lemma} follows once Lemma~\ref{prop:main-lemma-strong} is proven.

We first construct the map $\Pa$ and show (A), (C), and (D).
Proving (B) will be a bigger task.
We split this up into three sections: (1) the right-to-left inclusion, (2) the left-to-right inclusion in the case of small heights $n \leq N$, (3) the left-to-right inclusion on the case of large heights $n > N$.

\subsection{Distributive Approximators}

We first define a family of languages recognized by finite distributive forest algebras, acting as the right-most factors of the wreath products we want to build.
For each rule $\rho$, the language $\Delta(\rho)$ will be a distributive approximation of $\La(\rho)$, and should at least contain the trees belonging to $\La(\rho)$.
Intuitively, $\Delta(\rho)$ should be  the smallest language extending $\La(\rho)$ that is recognized by a finite distributive forest algebra.
Such a language would have to be a superset of $\pi^{-1}\pi\La(\rho)$.
However, such a smallest language does not in general exist.
For any fixed height limit $N$, we can find a finite distributive forest algebra which agrees with $\pi^{-1}\pi(\La(\rho))$ for forests of height up to $N$, but in general, no finite forest algebra can recognize $\pi^{-1}\pi(\La(\rho))$.
Our strategy will be to specify a height limit $N$ up to which $\Delta(\rho)$ will (in a certain sense) agree with $\pi^{-1}\pi(\La(\rho))$.
This height limit is specified by the following lemma:

\begin{llemma}[Constant]\label{lemma:absorption}

There is $N \in \N$ such that, for any choice of $s \in H_\Ff$, $S_1, S_2 \subseteq Rules(\Ff)$, 
and $n \geq N$, at least \textbf{one} of the following two statements holds: 

\begin{enumerate}
\item There is $f \in \La(s)$ such that $$\pi(f) \subseteq \Pi Z_n \left\langle S_1,S_2\right\rangle$$
\item For each $f \in \La(s)$, we have $$\pi(f) \cap (\Pi(\La(s)) - \Pi(\left\langle S_1,S_2\right\rangle)) \neq \emptyset$$
\end{enumerate}
\end{llemma}

\begin{proof}
We show that violation of (2) entails (1).

Fix $s, S_1, S_2$.
Assume $f \in \La(s)$ and $\pi(f) \cap (\Pi(\La(s)) - \Pi(\left\langle S_1,S_2\right\rangle)) = \emptyset$.
Since $\pi(f) \subseteq \Pi(\La(s))$, we also get $\pi(f) \subseteq \Pi(\left\langle S_1,S_2\right\rangle)$.
Due to horizontal imdepotency/commutativity, there is a forest $g \in \left\langle S_1,S_2\right\rangle$ such that $\pi(f) \subseteq \pi(g)$.
Let $N_{s,S_1,S_2} := \height(g)$.
Thus,  $(\dagger)$ $\pi(f) \subseteq \Pi Z_{N_{s,S_1,S_2}} \left\langle S_1,S_2\right\rangle$.

Since $H_\Ff$ is finite, we can take an integer $N := \max_{s,S_1,S_2} N_{s,S_1,S_2}$.
Furthermore, $Z_{N_{s,S_1,S_2}} \left\langle S_1,S_2\right\rangle \subseteq Z_n \La_{S_1,S_2}$ for any $s, S_1, S_2$.
In view of $(\dagger)$, if (2) is violated, then (1) holds.
\end{proof}


To define $\Delta$, we collect some properties expressible by distributive algebras that represent `minimal requirements' that any tree in $\La(\rho)$ would certainly fulfil.

\begin{defin}\label{defin:approximators}
\emph{(Distributive Approximators)} For $\rho \in Rules(A)$, define a set of trees $\Delta(\rho)$ as follows:

Let $t$ be a tree, with root symbol $\alpha$. Then $t \in \Delta(\rho)$ iff these conditions hold:

\begin{itemize}
\item The root symbol $\alpha$ is the symbol of the rule $\rho$. That is, $\rho$ has the form $\dots \leftarrow \alpha[\dots]$.

\item If $t$ has height $\leq N$, then there is a tree $t'$ of height $\leq N$ such that $\{t'\} \in \La(\rho)$ and $\pi t' \subseteq \pi t$

\item Take any $s \in H_\Ff$, $S_1, S_2 \subseteq Rules(\Ff)$, and assume that $$\pi(g) \cap \left(\alpha \circ \left(\Pi(\La(s)) -\Pi(\left\langle S_1, S_2\right\rangle)\right)\right)  \neq \emptyset$$ for any $g \in \La(\rho)$, where $\circ$ denotes concatenation: $\alpha \circ A := \{\alpha w : w \in A\}$.
Then, we have $$\pi(t) \cap \left(\alpha \circ \left(\Pi(\La(s)) -\Pi(\left\langle S_1, S_2\right\rangle)\right)\right)  \neq \emptyset$$

\end{itemize}
\end{defin}

At this point, let's point out the formal similarity between the third condition in this definition to the second condition in Lemma~\ref{lemma:absorption}.
We can think of the third condition as a coarse distributive approximation to $\La(\rho)$ that is as fine-grained as allowed by sets of the form $\Pi(\left\langle S_1, S_2\right\rangle)$. We will later use the definition of $\Pi(\left\langle S_1, S_2\right\rangle) $ and Lemma \ref{lemma:absorption} to see that this requirement, while quite weak, is useful at heights $> N$.

\subsection{Defining Approximators $\Pa$}

We now define the map $\Pa$. The idea is that languages $\Pa(R)$ approximate each $\pi\La(\rho)$ ($\rho \in Rules(\Ff_i)$) exactly up to depth $N$ given in Lemma \ref{lemma:absorption}, and up to the granularity provided by sets of the form $\Pi\left\langle R,S\right\rangle$ at greater depths.

\begin{defin}[Trails, Traces]
\emph{(Trails):}
Let $f$ a forest. A \emph{trail} in $f$ is a sequence of nodes starting at the root of one member tree, continuing towards the leaf until it ends. Trails do not have to be maximal.

\emph{(Traces):} Let $\Trace_\Ff \in Rules(\Ff,\phi)^*$ as follows: 
$\tau \in \Trace_\Ff$ iff, for all $1 \leq i < i+1 \leq |\tau|$, there is $s_i \in H_\Ff$ such that
$$\tau_i\ \text{has the form }\ \dots \leftarrow \dots \{\dots, s_i, \dots\}$$
$$\tau_{i+1}\ \text{ has the form }\ s_i \leftarrow \dots \{\dots\}$$

Given a path $w \in \Sigma^*$, we set $\Trace_\Ff(p)$ to be the set of traces $\tau$ where $|w| = |\tau|$ and the transition symbol of rule $\tau_i$ is equal to the symbol $w_i$, for all $i$.
\end{defin}

We're ready to define the approximators $\Pa$:

\begin{defin}[Approximators]\label{def:p}
Let $\rho \in Rules(\Ff)$. 
We say that a trail $p$ in a forest \emph{satisfies the conditions} for a rule $\rho$ if there is a trace $\zeta \in \Trace_{\Ff_i}$ such that (1) $|\zeta| = |p|$, (2) $\zeta_0 = \rho$, and (3) for each $j \in \{1, ..., |p|\}$, having the $j$-th node in $p$ as its root is in $\Delta(\zeta_j)$.

Now let $R \subseteq Rules(\Ff)$. Then $\Pa(R) \subseteq H_\Sigma$ is the language of nonempty forests where every trail satisfies the conditions for some $\rho \in R$. 

\end{defin}

Let us first note that Condition (D) of the Main Lemma is an instant consequence: As $\Pa$ is defined in terms of trails, it is enough to check membership of each singleton forest.

\subsection{Proving Condition (C) in the Main Lemma}
In order to prove Condition (C), which states recognition by wreath products of finite distributive algebras, we will need the operation $\Psi$ from Definition~\ref{def:psi}.
Using this construction, we can first establish the following result for $\Pa(R)$:

\begin{prop}
For each $R$, the language $\Pa(R)$ is recognized by a wreath product of two distributive finite forest algebras.
\end{prop}

\begin{proof}
We first want to show that, for any $\rho \in Rules(\Ff)$, the language $\Delta(\rho)$ is recognized by a finite distributive forest algebra.

Recall Proposition~\ref{prop:distr-char}.
In view of this, the first two conditions in the definition of $\Delta(\rho)$ can certainly be represented by a finite distributive forest algebra.

Now, let us consider the third condition.
$\Pi(s)$ is a regular language of words (Proposition~\ref{prop:pi-reg}).
To show that $\Pi(\left\langle S_1, S_2\right\rangle)$ is also regular, we can make use of $\Psi$ (Definition~\ref{def:psi}).
For $i = 1,2$, let $\La(S_i)$ be the language of forests $f_1 + ... + f_n$, where $f_j \in \La(\rho_j)$ with $\rho_j \in S_i$, for all $j$.
Given that $\Ff$ is finite, the languages $\La(S_1)$, $\La(S_2)$ are regular forest languages.
By definition of $\left\langle\cdot,\cdot\right\rangle$, we have $\left\langle S_1, S_2\right\rangle = \pi(\La(S_1)) \cap \pi(\La(S_2))$.
By Proposition~\ref{prop:psi}, we have $\pi(\Psi(\La(S_i))) = \pi(\La(S_i))$.
Thus, we can rewrite 
$\left\langle S_1, S_2\right\rangle = \pi(\Psi(\La(S_1))) \cap \pi(\Psi(\La(S_2)))$.
Now observe that, again due to Proposition~\ref{prop:psi}, we can permute $\pi$ outside of Boolean operations when they apply to images of $\Psi$. In our case,
$\pi(\Psi(\La(S_1))) \cap \pi(\Psi(\La(S_2))) = \pi(\Psi(\La(S_1)) \cap \Psi(\La(S_2))))$.
We can therefore write $\Pi(\left\langle S_1, S_2\right\rangle)$ as $\Pi(\Psi(\La(S_1)) \cap \Psi(\La(S_2)))$, which is a regular word language by Proposition~\ref{prop:pi-reg} and Proposition~\ref{prop:psi-reg}.
Thus, Proposition~\ref{prop:distr-char} also allows us to encode the third condition in a finite distributive forest algebra.

Since $\Ff$ is a finite forest algebra, there is a finite distributive algebra recognizing all languages $\Delta(\rho)$ ($\rho \in Rules(\Ff)$).
Checking whether a sequence of rules is a valid trace only requires testing pairs of adjacent symbols.
Therefore, $\Trace_\Ff$ is a regular word language.
Using the characterization of wreath products in terms of sequential composition (Theorem 4.2 of \cite{bojanczyk-wreath-2012}), it follows that $\Pa(R)$ is recognized by a wreath product of two finite distributive algebras.
\end{proof}

We can now prove condition (C) from the Main Lemma:
\begin{prop}
For any $r_1, ... r_k \in R$, the language $$\{f = f_1 + ... + f_k : f_i \in \Pa(\{r_i\})\}$$
is recognized by a wreath product of three finite distributive forest algebras.
\end{prop}

We show this as an instance of a general fact: 

\begin{prop}
Let $\La_1, ..., \La_n$ be languages closed under addition ($f, f' \in \La_i$ implies $f+f' \in \La_i$) and under taking nonempty subsets of member forests -- that is, for any $f \in \La_i$ and $\emptyset \subsetneq f' \subseteq f$ we have $f' \in\La_i$.

Let $\Ff := (H,V)$ be a forest algebra recognizing $\La_1, ..., \La_n$ via morphism $\phi$.
Then there is a finite distributive algebra $\Gg$ such that $\Gg \wr \Ff$ recognizes the language $\La_1 + ... + \La_n$.
\end{prop}

By definition of $\Pa$, each $\Pa(R)$ is closed under taking nonempty subsets of member forests: Since $\Pa(R)$ is defined in terms of traces, a forest $f$ is in $\Pa(R)$ if and only if all the trees in $f$ are in $\Pa(R)$.
Therefore, the previous proposition is an instant consequence of this fact.

\begin{proof}
Let $H_1$ be $Pow(Pow(\{1,...,n\}))$ with union as the monoid operation, and let $V_1 \subseteq H_1^{H_1}$ consist of (1) the set of constant functions $H_1 \rightarrow H_1$, and (2) the set of functions $f_K : L \mapsto K \cup L$ ($K, L \subseteq Pow(\{1,...,n\})$).
$V_1$ is a monoid, with function composition as the operation, and $f_\emptyset$ as the identity element.
$V_1$ acts on $H_1$ by $v\cdot h := v(h)$.
For $h \in H_1$, we can set $I_h$ to be $f_h$.
Thus, $(H_1,V_1)$ is a finite forest algebra.

Let's verify that $(H_1, V_1)$ is distributive:
Consider $v[h_1+h_2]$. If $v$ is a constant function, this is certainly equal to $vh_1+vh_2$.
Now consider the case where $v = f_K$. Then
$v[h_1+h_2] = K\cup h_1 \cup h_2 = (K\cup h_1) \cup (K \cup h_2) = v[h_1] + v[h_2]$.
We have shown that $(H_1, V_1)$ is distributive.

We claim that $(H_1, V_1) \wr \Ff$ recognizes the language $\La_1 + ... + \La_n$.

For each context type of the form $v = \alpha[X] \in V_\Sigma$ -- that is, consisting only of a variable and a parent node labeled $\alpha$, define a function
$f_\alpha : H \rightarrow V_1$ given by 
$$f_\alpha(h)(h') \equiv \{\{i : \alpha[\phi_\Ff^{-1}(h)] \cap \La_i \neq \emptyset\}\}$$
(independent of $h'$, given that $f_\alpha(h)$ is a constant function).
We can set $\eta(\alpha[X]) := (f_\alpha, \phi(\alpha))$ and can extend this map to a morphism $\eta : \Sigma^\Delta \rightarrow (H_1, V_1) \wr \Ff$.

For a forest $f$, which we can write as a sum of trees $t_1 + ... + t_n$, we have
$\eta(f) = \{\{i : t_j \in \La_i\} : j = 1, ..., n\}$.
Therefore, $\La_1 + .. + \La_n$ is the set of trees $f$ where (1) each element of $\eta(f)$ is nonempty, (2) each $i = 1, ..., n$ occurs in some element of $\eta(f)$.

Thus, $(H_1, V_1) \wr \Ff$ recognizes $\La_1 + ... + \La_n$ via $\eta$.
\end{proof}

\subsection{Correctness of Approximators}
We now show Condition (A) of the Main Lemma, which is to establish that the distributive approximators $\Delta(\rho)$ and the  approximators $\Pa(R)$ are `big' enough to include the intended languages:

\begin{prop}
Let $\rho \in Rules(\Ff)$. Let $t$ be a tree such that $\{t\} \in \La(\rho)$. Then $t \in \Delta(\rho)$.
\end{prop}

\begin{proof}
Directly from the definition of $\Delta(\rho)$.
\end{proof}

\begin{prop}[Condition A of the Main Lemma]\label{prop:correctness}
For $R \subseteq Rules(\Ff_i)$ and any $r_1, ..., r_n \in R$ ($n \geq 1$), we have 
$$\La(r_1) + ... + \La(r_n) \subseteq \Pa(R)$$
\end{prop}

\begin{proof}
We check this by going through the definition of $\Pa$.

Let $f \in \La(r_1) + ... + \La(r_n)$.
Let $p$ be a maximal trail (we only allow nonempty trails in the definition of `trail') in $f$.
Let $\tau \in (\operatorname{Rules}(\Ff))^{|p|}$ be the sequence of rules where $\tau_i$ is the rule that the tree rooted at the $i$-th element of $p$ evaluates to.
That is, $\tau \in \operatorname{Traces}_\Ff(p)$.
Certainly, $p$ must be running through some tree $t$ in $f$.
There is a rule $\rho = r_j \in R$ such that $t \in \La(\rho)$.
By choice of $\tau$, we have $\tau_0 = \rho \in R$.
Again, by choice of $\tau$, the tree rooted in the $i$-th element of $p$ is in $\La(\tau_i)$.
By the preceding proposition, it is also in $\Delta(\tau_i)$.

By the definition of $\Pa$, the claim follows.
\end{proof}

\subsection{Right-to-Left Inclusion for Condition (B)}

Up to this point, we have proven conditions (A), (C), (D) from the Main Lemma.
From Proposition~\ref{prop:correctness}, we can derive one of the two inclusions of condition (B) in the Main Lemma, which will be the goal of this section.
Recall $\Psi$ from Definition~\ref{def:psi}.

\begin{prop}\label{prop:transfer-psi}
Let $f$ be a forest.
If $f \in \Pa(R)$, then $\Psi f \in \Pa(R)$.
\end{prop}

\begin{proof}
Consider a trail in $\Psi f$ and some arbitrary trail consisting of the same symbol sequence in $f$.
For each node $\nu$ in $\Psi f$ along the trail, look at the corresponding node $\nu'$ along the trail selected in $f$.
For each node $\nu$, let $f_\nu, f_{\nu'}$ be the forests below $\nu, \nu'$, respectively.
Then $f_\nu$ is in the upward closure of $f_{\nu'}$ -- that is, we can obtain $f_\nu$ from $f_{\nu'}$ by adding nodes below existing nodes.
From the definition of $\Delta$, one verifies easily that $\Delta(\rho)$ is upward-closed for each $\rho$.
Also, membership in $\Delta(\rho)$ only depends on the set of paths.
Therefore, $f_{\nu'} \in \Delta(\rho)$ implies $f_\nu \in \Delta(\rho)$.
Then the claim follows from the definition of $\Pa(R)$.
\end{proof}

Now, we prove one of the two directions of the main lemma:

\begin{prop}[Right-to-Left Direction of the Main Lemma]\label{prop:one-direction}
\begin{equation}
\Pi \left(\Rn\Pa(R) \cap \Rn\Pa(S)\right) \supset \Pi\left\langle R, S\right\rangle 
\end{equation}
\end{prop}

\begin{proof}
Let $w \in \Pi \left\langle R, S\right\rangle$.
Then there is $f \in \left\langle R, S\right\rangle$ with $w \in f$.
Stated differently, there are $f_1, f_2$ such that $\pi(f_1) = \pi(f_2) = f$, $\height\ f_i \leq n$, $f_1 \in \La(r_1) + ... \La(r_k)$ and $f_2 \in \La(s_1) + ... + \La(s_j)$ ($r_m \in R, s_m \in S$), and $w \in \pi(f_i)$.

From Proposition \ref{prop:correctness}, we know $f_1 \in \Pa(R)$, $f_2 \in \Pa(S)$.
From Proposition \ref{prop:transfer-psi}, we conclude $\Psi f_1 \in \Pa(R)$, $\Psi f_2 \in \Pa(R)$.

Since $\pi(f_1) = \pi(f_2)$, we actually have $\Psi f_1 = \Psi f_2$ (Proposition~\ref{prop:psi}).
Thus, $\Psi f_i \in \Pa(R) \cap \Pa(S)$.
Given $\height\ f_i \leq n$, we have $\Psi f_i \in (\Rn\Pa(R) \cap \Rn\Pa(S))$.
In view of $w \in \pi(f) = \pi(\Psi f_i)$, we get $w \in \Pi (\Rn\Pa(R) \cap \Rn\Pa(S))$.
\end{proof}

\subsection{Left-to-Right Inclusion at Small Heights}

The goal of this section is to the prove the left-to-right inclusion in condition (B) in the Main Lemma in the case when $n \leq N$.

We say that a path $w \in \pi(f)$ \emph{belongs} to a trail $p$ if $|w|=|p|$ and the node $p_i$ is labeled $w_i$.

\begin{llemma}[Technical Lemma]\label{lemma:existsWitnessTree}
Let $n \leq N$, $\rho \in Rules(\Ff)$.
Take $t$ a tree.
Let $w \in \pi(t)$, belonging to trail $p$.
Assume there is a trace $\tau \in Traces_\Ff(w)$ such that the tree rooted in $\pi_i$ is in $\Delta(\tau_i)$, for each $i = 1, ..., |\pi|$.

Then there is a tree $g$ such that
\begin{enumerate}
\item $g \in \La(\tau_1)$
\item $w \in \pi(g) \subseteq \pi(t)$
\end{enumerate}
\end{llemma}

\begin{proof}
We do induction over the structure of the tree $t$.

The tree $t$ has the form $\alpha[\{t_1, ..., t_k\}]$, where $k \geq 0$ indexes the children. (In the base case, $k=0$ and there are no children). By induction hypothesis, we can assume the lemma holds for each of the trees $t_i$.

By assumption, $t \in \Delta(\tau_1)$.
Since $t$ has height $\leq N$, there is a tree $t' \in Z_N(\La(\tau_1))$ such that $\Pi(t') \subseteq \pi(t)$ according to the definition of $\Delta(\tau_1)$. Note that the root of $t'$ must be $\alpha$, since there are no empty trees.


If $|\tau| = 1$, set $g := t'$, and we have $g = t' \in \La(\tau_1)$ and $w = \alpha \in \pi(t') \subseteq \pi(t)$.

Now assume $|\tau| > 1$.
Then one of the children $t_i$ is rooted in $\tau_1$, and $\alpha^{-1}w \in \pi(t_i)$, belonging to trail $p_{1\dots|p|}$.
Also, $\tau' := \tau_{1\dots|p|}$ is a trace and the tree rooted in $\tau'_j$ is in $\Delta(\tau'_j)$ for each $j = 1,...,|\tau'|$.

Thus, we can apply the induction hypothesis and obtain a tree $g'$ such that $g' \in \La(\tau_2)$ and $\alpha^{-1}w \in \pi(g') \subseteq \pi(t_i)$.

Let $t'_1, ..., t'_l$ be the children of $t'$ ($l \geq 0$, as the children set may be empty). Now set $g := \alpha[\{t'_1, ... t'_l, g'\}]$.
Then $w \in \alpha\pi(g') \subseteq \pi(g)$. This shows (2).

Recall $g' \in \La(\tau_2)$. The rule $\tau_2$ evaluates to a forest type $q$ that occurs on the right-hand side of $\tau_1$.
Thus, since $t' \in \La(\tau_1)$, there is a child $t'_r \in \La(q)$.
Due to horizonal idempotency, $(t'_r + g') \in \La(q)$.
Thus, $\{t'_1, ... t'_l, g'\}$ evaluates to the same set of forest types $\subseteq H_\Ff$ as $\{t'_1, ... t'_l\}$.
Therefore, $g \in \La(\tau_1)$.
This proves (1).
\end{proof}

\begin{llemma}\label{lemma:up-to-n-pre}
Let $n \leq N$ and $R \subseteq Rules(\Ff)$. Then
$$\pi \Rn (\Pa(R)) \subseteq \left\langle R, R\right\rangle$$

\end{llemma}

\begin{proof}
Take $f \in Z_n \Pa(R)$.
Let $w \in \pi(f)$.
By definition of $\Pa$, we obtain a trail $\pi$ in $f$ belonging to $w$ and a trace $\tau$ such that $\tau_1 \in R$ and the tree rooted at $\pi_i$ is in $\Delta(\tau_i)$, for $i = 1, ..., |\pi|$.

Let $t_w$ be the trail rooted at $\pi_1$.
From Lemma~\ref{lemma:existsWitnessTree}, we obtain a tree $t'_w$ such that $t'_w \in \La(\tau_1)$ and $w \in \Pi(t'_w) \subseteq \pi(t)$.
Define a forest $g := \{t'_w : w \in \pi(f)\}$.
By construction, we have $\pi(f) = \pi(g)$, and $g \in \La(r_1) + ... + \La(r_k)$ ($r_j \in R$).
Thus, $\pi(f) = \pi(g) \in \left\langle R, R\right\rangle$.
\end{proof}

\begin{corollary}\label{lemma:up-to-n}
Let $n \leq N$ and $R,S \subseteq Rules(\Ff)$. Then
$$\Pi  (\Rn\Pa(R) \cap \Rn\Pa(S)) \subseteq \Pi \Rn \left\langle R, S\right\rangle$$
\end{corollary}

\begin{proof}
Let $w \in \Pi  (\Rn\Pa(R) \cap \Rn\Pa(S))$,
so there is $f \in (\Rn\Pa(R) \cap \Rn\Pa(S))$ with $w \in \pi(f)$.
That is, $\pi(f) \in \pi(\Rn\Pa(R))$ and $\pi(f) \in \pi(\Rn\Pa(S))$.
From Lemma \ref{lemma:up-to-n-pre}, we have $\pi(f) \in \left\langle R, S\right\rangle$.
Since $w \in \pi(f)$, $w$ is in the path set of that language, which proves the claim.
\end{proof}

\subsection{Left-to-Right Inclusion at Large Heights}

To complete the Main Lemma, we prove the left-to-right inclusion of condition (B) for arbitrary heights $n$.

\begin{defin}[Quotients]\label{def:quotient}
If $f$ is a forest, we write $\alpha^{-1}f := \bigcup_{f' : \alpha[f'] \in f} f'$.
Thus, $\alpha^{-1}f$ is also a forest, consisting of all the trees occurring below roots labeled $\alpha$ in members of $f$.

If $\La$ is a forest language, we write $\alpha^{-1}\La := \{\alpha^{-1}f : f \in \La\}$.
\end{defin}

\begin{llemma}[Second Direction at Arbitrary Heights]\label{main-lemma-final}
Let $R,S \subseteq Rules(\Ff)$, and $n \in \N$. Then
$$\Pi \left(\Rn\Pa(R) \cap \Rn\Pa(S)\right) = \Pi \Rn\left\langle R, S\right\rangle$$
\end{llemma}

\begin{prop}
It suffices to prove Lemma~\ref{main-lemma-final} in the case where all rules in $R, S$ have the same transition symbol $\alpha$.
\end{prop}

\begin{proof}
Let any sets $R, S \subseteq Rules(\Ff)$ be given.
We can partition the set $R$ into sets $R_\alpha$ ($\alpha \in \Sigma$) according to the transition symbol, and similarly for $S$.

Recall that our goal is to show the $\supseteq$ direction, 
as the $\subseteq$ direction is already known from Proposition~\ref{prop:one-direction}.

So let $w \in \Pirn \left(\left(\Rn\Pa\left(R\right) \cap \Rn\La\left(S\right)\right)\right)$.
Then there is a tree $t$ such that $\{t\} \in \left(\left(\Rn\Pa\left(R\right) \cap \Rn\La\left(S\right)\right)\right)$  and $w \in \pi(t)$.
Let $\alpha := w_1$ be the root symbol of $t$.
In view of the definition of $\Pa$ and $\Delta$, we can restrict the statement to rules with this symbol: $\{t\} \in \left(\Rn\Pa(R^\alpha) \cap \Rn\Pa\left(S^\alpha\right)\right)$.
Now let us presume that the lemma has already been shown for the pair $R^\alpha, S^\alpha$.
Applying this, we get
$$w \in \Pirn \left(\Pa\left(R^\alpha\right) \cap \Pa\left(S^\alpha\right)\right) = \Pirn \left\langle R^\alpha, S^\alpha \right\rangle \subseteq \Pirn \left\langle R, S \right\rangle$$
\end{proof}

The remainder of this section will be devoted to completing the proof of Lemma~\ref{main-lemma-final} under the assumption that all rules in $R, S$ have the same transition symbol $\alpha$. This will then complete condition (B) of the Main Lemma.

We only need to show $\subseteq$, the other direction following from Prop \ref{prop:one-direction}.
We do induction over $n$ and prove this for all sets of rules simultaneously. We have already shown this for $n \leq N$ in Lemma \ref{lemma:up-to-n}, which serves as the inductive base.
So, for the inductive step, let $n > N$.
During the inductive step, we will assume that sets $R, S$ of rules are given where all rules have the same transition symbol $\alpha$.
Applying the previous proposition then completes the inductive step for all sets $R, S$.

Giving names to the two sides in the Lemma~\ref{main-lemma-final}, the goal is to prove equality between the path set of the language
\begin{equation}
\phi :=  \left(\Rn\Pa\left(R\right) \cap \Rn\Pa\left(S\right)\right)
\end{equation}
and the (finite) path set
\begin{equation}\label{def:g}
g := \Pi \Rn \left\langle R, S\right\rangle 
\end{equation}

\emph{Note that $\phi$ is a finite (due to $\Rn$) forest language, while $g$ is a finite set of words.}
That is, the goal of this section is to prove
$$\Pi \phi = g$$
We have $\Pi \phi \supseteq g$ by Lemma \ref{prop:one-direction}.

We write $R = \{R_1, ..., R_{n_1}\}$, and $S = \{R_1, ..., R_{n_2}\}$.
We write (recall Definition~\ref{def:algebras})
\begin{equation}\label{eq:states-right}
\begin{split}
R_i &= \left\langle s_0^i, \alpha, \{s^i_l : l = 1, ..., n_{R_i}\} \right\rangle \\
S_i &= \left\langle t_0^i , \alpha ,\{t^i_l : l = 1, ..., n_{S_i}\}\right\rangle
\end{split}
\end{equation}

When $s \in H_\Ff$, define 
\begin{equation}\label{def:rho-state}
\rho(s) := \{ R \in Rules(\Ff) : \pi^{(1)}(R) = s\}
\end{equation}

To apply the induction hypothesis, we consider
\begin{equation}\label{def:hat-phi}
\widehat{\phi} := \Rnm\Pa\left(\bigcup_{i=1}^{n_1} \bigcup_l \rho(s^i_l)\right) \cap \Rnm\Pa\left( \bigcup_{j=1}^{n_2} \bigcup_l \rho(t^j_l)\right)
\end{equation}

Let us reflect what this means: Instead of the rules in $R$, we take all the rules that evaluate to a forest type that appears on the right-hand side of a rule in $R$. (The same for $S$.)
The language 
$\Pa\left(\bigcup_{i=1}^{n_1} \bigcup_l \rho(s^i_l)\right)$
is a basic  approximation to the quotient of $\Pa(R)$ by the root symbol $\alpha$.
Note that, as removing the root symbol reduces height, we're using height $n-1$ in the definition of $\widehat{\phi}$.
Therefore, the object $\widehat{\phi}$ is a basic  approximation to the quotient of $\phi$ by the root symbol $\alpha$. Since it has smaller height, we can apply the induction hypothesis.
However, it might be larger than the quotient of $\phi$ by $\alpha$, so we'll need to do extra work before we can use the induction hypothesis to say something about $f$. This is the motivating idea behind the remainder of the proof, where we `shrink' $\widehat{\phi}$ until its pathset is equal to that of the quotient of $\phi$ by the root symbol $\alpha$.

Here it might be useful to remark that $\Pa$ and $\cup$ do not commute. For instance, $\Pa(\{\rho_1\}) \cup \Pa(\{\rho_2\}) \subseteq \Pa(\{\rho_1, \rho_2\})$, but the converse usually does not hold. The reason is that $\Pa$ looks at each trail separately and cannot `keep track' of which tree `belongs' to which rule.

Returning to the proof, since $\height(\widehat{\phi}) = n-1 < n$, we can plug $\widehat{\phi}$ into the induction hypothesis, and obtain
\begin{equation}\label{ex:ih}
\Pi\widehat{\phi} =^{(IH)} \Pi  Z_{n-1} \left\langle \left(\bigcup_{i=1}^{n_1} \bigcup_l \rho(s^i_l)\right) ,\left(\bigcup_{j=1}^{n_2} \bigcup_l \rho(t^j_l)\right)\right\rangle 
\end{equation}


As the next proposition shows, we can view $\Pi \widehat{\phi}$ as an outer approximation to $\alpha^{-1} g$ (recall Definition~\ref{def:quotient}) and $\alpha^{-1} \Pi \phi$:

\begin{prop}\label{obs:phi-phi-hat}
$$\alpha^{-1} g \subseteq \alpha^{-1} \Pi \phi \subseteq \Pi \widehat{\phi}$$
\emph{Note that $\alpha^{-1}$ refers to quotients of word languages here, not forest languages.}
$$\alpha^{-1} \phi \subseteq \widehat{\phi}$$
\emph{Note that $\alpha^{-1}$ refers to the quotient of a forest language here.}
\end{prop}


\begin{proof}
We already know the first inequality of the first claim since we know $g \subseteq \Pi \phi$ from Proposition \ref{prop:one-direction}.
The second inclusion of the first claim follows immediately from the second claim. Let's now prove the second claim:

Let $f \in \alpha^{-1} \phi$.
By definition of quotients of forest languages (Definition~\ref{def:quotient}), there is $f' \in \phi$ and $t_1, ..., t_k \in f'$ such that $t_i = \alpha[f_i]$ ($i = 1, ..., k$) and $f_1 + ... + f_k = f$.
Now we want to appeal to the definition of $\Pa$ (Definition~\ref{def:p}).
So let $p$ be a trail in $f$, so it is also a trail in some $f_i$.
This can be extended to a trail $\nu p$ in $t_i$ (where $\nu$ is the root of $t'$).
Since $t_i \in \phi$, we know $\{t_i\} \in \Pa(R)$.
By definition of $\Pa$, there is a trace $\tau \in \Trace_\Ff$ such that (1) $|\tau| = |p|+1$, (2) $\tau_1 \in R$, say $\tau_1 = R_k$, (3) for each $j \in \{1, ..., |p|+1\}$, the tree rooted at $(\nu p)_j$ is in $\Delta(\tau_j)$.
Now by the definition of $\Trace_\Ff$ and since $\tau_1 = R_k$, we know $\tau_2 \in \bigcup_l \rho(s^k_l) \subseteq \bigcup_{i=1}^{n_1} \bigcup_l \rho(s^i_l)$.
So let $\tau' := \tau_{2...|\tau|} \in \Trace_\Ff$. Then we have (1) $|\tau'| = |p|$, (2) $\tau'_1 \in \bigcup_{i=1}^{n_1} \bigcup_l \rho(s^i_l)$, (3) for each $j \in \{1, ..., |p|\}$, the tree rooted at $p_j$ is in $\Delta(\tau'_j)$.
By definition of $\Pa$, we know $f \in \Pa\left(\bigcup_{i=1}^{n_1} \bigcup_l \rho(s^i_l)\right)$.

The proof for the $S$-part is analogous. Together, we get $f \in \widehat{\phi}$.
\end{proof}

We now want to prove that $\Pi \phi = g$, which as explained above is exactly the statement of the Main Lemma to be proven.
We construct monotonically decreasing sequences $\widehat{\phi}_k \subseteq \freedisth$ and $W^1_k \subseteq R$, $W^2_k \subseteq S$ ($k = 0, 1, ...$).
Our proof strategy will be to `sandwich' $\Pi\phi$ between this decreasing sequence and $g$.


To make notation nicer, we will write (recall (\ref{def:hat-phi})):
\begin{align*}
S_k^1 := \bigcup_{i \in W^1_k} \bigcup_l \rho(s^i_l) \\
S_k^2 := \bigcup_{i \in W^2_k} \bigcup_l \rho(t^i_l)
\end{align*}

\begin{defin}
Here is the inductive definition:
\begin{itemize}
\item $\widehat{\phi}_0 := \widehat{\phi}$, $W^1_0 := R$, $W^2_0 := s$ 

\item  We are in the $k+1$-th step.

\begin{enumerate}
\item \emph{Case I}: For each $\rho \in W^1_k \cup W^2_k$, there is a forest $h \in \La(\rho)$ such that $\pi(h) \subseteq \left(\alpha \Pi \widehat{\phi}_k \cup \{\epsilon\}\right)$.

We take $$W^s_{k+1} := W^s_k$$ for $s = 1,2$.

\item \emph{Case II}: Case I is not satisfied. That is, there is $\rho \in W^z_k$ ($z \in \{1,2\}$), such that, for each $h \in \La(\rho)$ we have $\pi(h) \not\subseteq \alpha \Pi \widehat{\phi}_k \cup \{\epsilon\}$.

For definiteness, we take $z = 1$. So $\rho \in W^1_k$. 
Set $$W^1_{k+1} := W^1_k - \{\rho\},\ \ W^2_{k+1} := W^2_k$$
The construction is symmetric if instead $z=2$.
\end{enumerate}
This concludes the inductive definition of $W^1_k, W^2_k$. In either of the two cases, we set $$\widehat{\phi}_{k+1} := \widehat{\phi}_{k} \cap \left(\Pa(S^1_{k}) \cap \Pa(S^2_{k})\right)$$

\end{itemize}
\end{defin}

Since $|W^1_k| + |W^2_k|$ gets strictly smaller whenever we are in Case II, we can reach Case II only finitely often (If the sets happen to become empty, Case I becomes trivially true.), and the sequences $W^1_k, W^2_k$ become stationary.

Having defined the sequences $\widehat{\phi}_k$, $W^i_k$, we first note the following: 

\begin{prop}\label{b:second}
\begin{enumerate}
\item  For all $k$, we have $$\widehat{\phi}_{k+1} = \left(\Rnm\Pa(S^1_k) \cap \Rnm\Pa(S^2_k)\right)$$
\item \label{b:seventh} $\Pi \widehat{\phi}_{k+1} = \Pi  \Rnm \left(\left\langle S^1_k, S^2_k    \right\rangle\right)$
\end{enumerate}
\end{prop}

\begin{proof}
\begin{enumerate}
\item is immediate from the definition.

\item follows from (a) and the induction hypothesis, analogous to (\ref{ex:ih}).

\end{enumerate}
\end{proof}

Recall that in Proposition~\ref{obs:phi-phi-hat} we established that $\widehat{\phi}$ extends $\alpha^{-1} \phi$.
We now show that this remains true for all $\widehat{\phi}_k$.
This is the place where the choice of $N$ and $\Delta(\rho)$ come in crucially.

\begin{prop}\label{b:first}
\begin{equation}
 \forall k : \alpha^{-1} \phi \subseteq \widehat{\phi}_k
\end{equation}

\end{prop}

\begin{proof} By induction over $k$. In the case of $k=0$, we established this in Proposition~\ref{obs:phi-phi-hat}.

In the inductive step, we want to show $\widehat{\phi}_k \supset \alpha^{-1} \phi$ and already know $\widehat{\phi}_{k-1} \supset \alpha^{-1} \phi$.
Assume $\widehat{\phi}_k \not\supset \alpha^{-1} \phi$.
So there is a forest
\begin{equation}
f \in \left(\alpha^{-1} \phi  - \widehat{\phi}_k\right) = \left((\alpha^{-1} \phi \cap \widehat{\phi}_{k-1}) - \widehat{\phi}_k\right)
\end{equation}

We know that $f$ is in $\widehat{\phi}_{k-1} - \widehat{\phi}_k$, which by Proposition~\ref{b:second} is equal to
$$Z_{n-1} \left[(\Pa(S^1_{k-2}) \cap \Pa(S^2_{k-2})) - (\Pa(S^1_{k-1}) \cap \Pa(S^2_{k-1})) \right]$$

We have $f \in \alpha^{-1} \phi$. This means there is a forest $f_0 \in \phi$ such that $f$ is the collection of the children of the root nodes in $f_0$ (all root nodes in $\phi$ have label $\alpha$).

Recall that, by definition of $\Pa(R)$ (Definition~\ref{def:p}), it is the language of forests where each trail satisfies the conditions for some $\rho \in R$.

From the definition of $\Pa$ and $\phi \in \Pa(R)$, we know that each maximal trail in $f_0$ satisfies the conditions for some rule in $R$.
Also, there is a trail $p'$ in $f$ that does not satisfy the condition for any forest type in $S^1_{k-1}$.
 This trail extends to a trail  $p$ in $f_0$ that does not satisfy the condition for any rule in $W^1_{k-1}$ (+).
The trail $p$ satisfies the conditions for some rule $\rho \in R$, since $p \in \phi$ and $\phi \in \Pa(R)$.
 So the trail $p$ satisfies the conditions for some rules in $R - W^1_{k-1}$ but for no rule in $W^1_{k-1}$.
From (+), we know $\rho \in R - W^1_{k-1}$.
So $\rho$ was removed at some prior stage $k' < k-1$ of the construction.

So, in stage $k'$, Case II was reached, and for every $h \in \La(\rho)$, we have $\pi(h) \not\subseteq \alpha\Pi\widehat{\phi}_{k'} \cup \{\epsilon\}$.
Then there is a forest type $s_l^q$ from the third component of $\rho$ such that, for any $h \in \La(s_l^q)$, we have $h \not\subseteq \Pi\widehat{\phi}_{k'}$.

From Proposition~\ref{b:second}.2, we have
\begin{equation}\label{eq:c-conc}
\Pi \widehat{\phi}_{{k'}} = \Pi \Rnm  \left\langle S^1_{{k'}-1}, S^2_{{k'}-1}\right\rangle \subseteq \Pi\left\langle S^1_{{k'}-1}, S^2_{{k'}-1}\right\rangle 
\end{equation}

Now take $I := \Pi(\La(s_l^q)) - \Pi\left\langle S^1_{{k'}-1}, S^2_{{k'}-1}\right\rangle $.
By (\ref{eq:c-conc}), (*) $I \cap \Pi \widehat{\phi}_{{k'}} = \emptyset$.
At this point, we can apply Lemma \ref{lemma:absorption}:
Since $n \geq N$, for each $h \in \La(s_l^q)$, we have $\pi(h) \cap I \neq \emptyset$.
Therefore, for any $h \in \La(\rho)$, we have $\pi(h) \cap (\alpha \circ I) \neq \emptyset$.
From the third point in the definition of $\Delta$ (Definition~\ref{defin:approximators}), we conclude: ($\dagger$) For each $h \in \Delta(\rho)$, we have $\pi(h) \cap (\alpha \circ I) \neq \emptyset$.

We know $p$ satisfies the conditions for $\rho$, so $\pi(f_0) \in \Delta(\rho)$.
Since $f_0 \in \phi$, we get $\pi(\phi) \cap (\alpha \circ I) \neq \emptyset$ from ($\dagger$).
As we have $\widehat{\phi}_{k-1} \supset \alpha^{-1} \phi$ from the inductive hypothesis, we obtain $\pi(\widehat{\phi}_{k-1}) \cap (\alpha \circ I) \neq \emptyset$.
On the other hand, $k' < k-1$ and thus $\widehat{\phi}_{k'} \supset \widehat{\phi}_{k-1}$.
This is a contradiction to $(*)$. Thus, the initial assumption was incorrect.
\end{proof}
Putting together what we have shown so far, we get a chain of inequalities:
\begin{equation}
\alpha^{-1} g \subseteq \alpha^{-1} \Pi \phi \subseteq \dots \subseteq \Pi\widehat{\phi}_{3} \subseteq \Pi\widehat{\phi}_{2} \subseteq \Pi\widehat{\phi}_{1} \subseteq \Pi\widehat{\phi}
\end{equation}
In order to sandwich $\alpha^{-1} \Pi \phi$ between $\alpha^{-1} g$ and the decreasing sequence $\Pi\widehat{\phi}_{k}$, we show that the sequence ultimately takes on the value $\alpha^{-1} g$.
First, we show

\begin{prop}\label{obs:large-k}
For $k$ such that Case II is never reached for $k' \geq k-1$, we have 
\begin{equation}
 \label{b:fk-g} 
\alpha (\Pi \widehat{\phi}_{k}) \subseteq g
\end{equation}
where $\alpha (\Pi \widehat{\phi}_k) := \{\alpha w : w \in \Pi\widehat{\phi}_k\}$

Since Case II is reached only finitely many often, there in particular is such a $k$.

\end{prop}

\begin{proof}

By assumption, for each $\rho \in W^1_k$, there is $\phi \in \La(\rho)$ such that $\Pi(\phi) \subseteq \Pi\widehat{\phi}_k$.
Also, for each $\rho' \in W^2_k$, there is $\phi \in \La(\rho')$ such that $\Pi(\phi) \subseteq \Pi\widehat{\phi}_k$.

Therefore, for each $\rho \in W^1_k$ and each $\rho'$, there is $\phi \in \La(s^\rho_l)$ such that $\Pi(\phi) \subseteq \alpha^{-1}\Pi\widehat{\phi}_k$.
And same for $W^2_k, t^i_l$. ($\dagger$)

By Proposition \ref{b:second} (\ref{b:seventh}), we know
\begin{equation}\label{eq:inter}
\alpha \Pi\widehat{\phi}_k = \alpha \Pi \Rnm \left\langle \left(\bigcup_{\rho \in W^1_k} \bigcup_l s^\rho_l \right),\left( \bigcup_{\rho' \in W^2_k} \bigcup_j t^{\rho'}_j\right)\right\rangle
\end{equation}

Now let $d \in \Rnm \left\langle \bigcup_{\rho \in W^1_k} \bigcup_l s^\rho_l , \bigcup_{\rho' \in W^2_k} \bigcup_j t^{\rho'}_j\right\rangle$,
so there are forest types $\sigma_j := s^{\rho}_l$ ($\rho \in W^1_k$) such that
$d = \pi(f_1 + ... + f_k)$
with $f_j \in \La(\sigma_j)$.
Using $(\dagger)$, for each $\rho \in W^1_k$ and each $j$, there is $\phi_{\rho,j} \in \La(s^\rho_l)$ such that $\Pi(\phi) \subseteq \alpha^{-1}\Pi\widehat{\phi}_k$.
If we set $d' := d \cup \pi(\sum_{\rho,j} \phi_{\rho,j})$, then we have $d \subseteq d'$, but also
$$d' \in\Rnm \left(\left( \sum_{\rho \in W^1_k} \sum_l\pi \La\left(s^\rho_l\right)\right) \cap  \left(\sum_{\rho' \in W^2_k} \sum_j\pi \La\left(t^{\rho'}_j\right)\right)\right)$$
Thus, $d \subseteq d' \subseteq \Pi\Rnm \left( \sum_{\rho \in W^1_k} \sum_l\pi \La\left(s^\rho_l\right) \cap  \sum_{\rho' \in W^2_k} \sum_j\pi \La\left(t^{\rho'}_j\right)\right)$

Putting this together with (\ref{eq:inter}), we obtain:
$$\alpha \Pi\widehat{\phi}_k = \alpha \Pi \left( \sum_{\rho \in W^1_k} \sum_j\pi\Rnm \La\left(s^\rho_j\right) \cap  \sum_{\rho' \in W^2_k} \sum_j\pi\Rnm \La\left(t^{\rho'}_j\right)\right)$$

Considering how paths of forests are built from paths of their children, we can rewrite this as
$$\left(\alpha \Pi\widehat{\phi}_k\right) \cup \{\epsilon\} =\Pi  \left(\sum_{\rho \in W^1_k} \pi\Rn\La(\rho) \cap \sum_{\rho' \in W^2_k} \pi\Rn\La(\rho')\right)$$

In view of the definition of $\left\langle \cdot,\cdot\right\rangle$, the right-hand side is a subset of $\Pi\Rn\left\langle W^1_k, W^2_k\right\rangle$. 
We therefore get
$$\alpha \Pi\widehat{\phi}_k \subseteq \Pi \Rn \left\langle W^1_k, W^2_k\right\rangle$$
Now $W^1_k \subseteq R$, $W^2_k \subseteq S$, and thus
$$\alpha \Pi\widehat{\phi}_k \subseteq \Pi \Rn \left\langle R, S\right\rangle$$
The expression on the right side is equal to $g$ by definition of $g$ (\ref{def:g}).

Taken together, we have shown $\alpha\Pi\widehat{\phi}_k \subseteq g$.
\end{proof}

As a converse to the last observation, we have:
\begin{prop}
$$g \subseteq \Pi \phi \subseteq \left(\{\epsilon\} \cup (\alpha \Pi \widehat{\phi}_k)\right)$$
\end{prop}

\begin{proof}
We have shown the first inclusion, $g \subseteq \Pi \phi$, previously (Proposition~\ref{prop:one-direction}).
Let's consider the second one, $\Pi \phi \subseteq \left(\{\epsilon\} \cup (\alpha \Pi \widehat{\phi}_k)\right)$.
We already know $\alpha^{-1}\phi \subseteq \widehat{\phi}_k$ (Lemma~\ref{b:first}).
This entails $\Pi(\alpha^{-1}\phi) \subseteq \Pi\widehat{\phi}_k$.
Since every tree occurring in elements of $\phi$ has $\alpha$ as its root node symbol, we have $\Pi(\phi) = \{\epsilon\} \cup \alpha\Pi(\alpha^{-1}\phi)$.
Thus, we conclude $\Pi \phi \subseteq \{\epsilon\} \cup (\alpha \Pi \widehat{\phi}_k)$.
\end{proof}

From the last two Propositions, for large $k$, we get a chain of inclusions
\begin{equation}\label{eq:chain}
\alpha\Pi\widehat{\phi}_k \subseteq g \subseteq \Pi \phi \subseteq \left(\{\epsilon\} \cup (\alpha \Pi \widehat{\phi}_k)\right)
\end{equation}
Considering $\epsilon \in g$ (any pathset must contain $\epsilon$), we deduce $$\left(\{\epsilon\} \cup (\alpha\Pi\widehat{\phi}_k)\right) = g$$
and, looking at the inclusion chain (\ref{eq:chain}) again $$g = \Pi \phi$$
This concludes, first, the inductive step, and thus the entire proof of the Main Lemma.

\newpage
\section{Proof of Proposition~\ref{prop:2-wreath}}\label{sec:proof-2-wr}

Let $\Ff$ be two-distributive, and let $\phi : \Sigma^\Delta \rightarrow \Ff$ be a morphism.
In Definition~\ref{def:pathsets}, we defined $\freedisth$ as the image of $\pi$ and thus a set of pathsets.
However, we can also view it as a monoid of forest types:
Let $\sim$ be the congruence on $H_\Sigma$ defined by $f \sim f' \Leftrightarrow \pi(f) = \pi(f')$.
This extends to a congruence on $V_\Sigma$.
The quotient of $\Sigma^\Delta = (H_\Sigma, V_\Sigma)$ by this congruence is an infinite forest algebra whose forest types are precisely the elements of $\freedisth$.
It is also not hard to show that this infinite forest algebra is equal to the quotient of $\Sigma^\Delta$ by the congruence induced by $v[h+h'] = vh+vh'$, and is therefore distributive.
We will call this infinite distributive forest algebra $\Sigma^\Delta_D = (\freedisth, V_\Sigma^D)$.
Its forest types are equivalence classes of forests that have the same path set, and can thus be represented as finite path sets.
We can extend $\pi$ to a forest algebra morphism $\Sigma^\Delta \rightarrow \Sigma^\Delta_D$, which is the canonical projection for the congruence $\sim$.

The proof of the Proposition now closely follows the proof of the main theorem.
Analogous to that proof, let us consider $D_{\phi, \pi}$, which has \emph{infinitely} many objects.
Note that the proof of the Derived Category Theorem (Theorem~\ref{thm:derived-cat}) given by \cite{straubing-forest-2018} applies even when the involved algebras and categories are infinite.
We want to show that $D_{\phi,\pi}$ is locally distributive.
Let $h, h'$ be objects, let $f_1, f_2 \in \HArr(h)$, and let $v \in \Arr(h,h')$.
By the definition of the Derived Category, we can write $f_1$ as $\carrow{}{h_1}h$ and $f_2$ as $\carrow{}{h_2}h$.
Also, we can write $v$ as $\carrow{h}{p}{h'}$, with $p \in V_\Ff$.
We want to prove the equality from Definition~\ref{def:loc-dist}.
Note that $h_1, h_2 \in H_\Ff$.
In view of the construction of the half-arrows in the derived category, there are forests $t_1, t_2 \in H_\Sigma$ such that $\phi(t_i) = h_i$ and $\pi(t_i) = h$ for $i=1,2$.
Let $\alpha' \in \phi^{-1}(p)$.
Since $\Ff$ is 2-distributive, we have $\phi(\alpha'[t_1+t_2]) = \phi(\alpha' t_1 + \alpha' t_2)$.
Applying $\phi$, this means $$p(h_1+h_2) = p(h_1) + p(h_2)$$
In the derived category, this translates to $$v (f_1 + f_2)) = v f_1 + v f_2$$
or, in arrow-based notation,
$$ \carrow{(\carrow{}{h_1}h + \carrow{}{h_2}h)}{p}{h'} =  \carrow{\carrow{}{h_1}h}{p}{h'} +  \carrow{\carrow{}{h_2}h}{p}{h'}$$
Thus, $D_{\phi,\pi}$ is locally distributive.

It is not finite, so we cannot directly apply the Local-Global Theorem~\ref{thm:loc-glob} here.
If we inspect the proof of the Local-Global Theorem, we see that the first part still applies here:
For any two distinct half-arrows in $D_{\phi, \pi}$, we get $$(\dagger)\ \pi_{\Sigma'}(\operatorname{Val}^{-1}(\carrow{}{c}h)) \cap \pi_{\Sigma'}(\operatorname{Val}^{-1}(\carrow{}{c'}{h'})) = \emptyset$$
where $\pi_{\Sigma'}$ maps forest-diagrams to sets of paths over the alphabet $\Sigma'$ consisting of arrows and half-arrows of $D_{\phi,\pi}$.
This alphabet is \emph{infinite}, but this needn't concern us:
$\Sigma'^\Delta$ is again a free forest algebra.
Furthermore, $\Sigma'^\Delta_D$, the quotient by the congruence induced by $\pi_{\Sigma'}$, is an infinite distributive algebra.
We now construct a division $D_{\phi,\pi} \prec \Sigma'^\Delta_D$ analogous to our reasoning in the proof of the Local-Global Theorem.
We assign to each half-arrow $\carrow{}{c}h$ in $D_{\phi,\pi}$ the set of half-arrows in $\Sigma'^\Delta_D$ that are in the image of $\operatorname{Val}^{-1}(\carrow{}{c}h)$ under $\pi_{\Sigma'}$, and similarly for arrows.
As in that proof, this mapping preserves operations, and, due to $(\dagger)$, is injective. 

In view of $D_{\phi,\pi} \prec \Sigma'^\Delta_D$, the Derived Category Theorem now implies that $\Ff$ divides the wreath product of the two infinite distributive algebras $\Sigma'^\Delta_D$ and $\Sigma^\Delta_D$.

\newpage
\section{Proof for Example~\ref{ex:langs}}\label{sec:proof-ex}

To show that the syntactic algebra is 2-distributive, we first note that $\pi$ separates $\La_1$ and $\La_2$:
We say that $f$ is \emph{compatible with $\La_1$} if \[\max_n \left(a^nb \in \pi(f)\right) = \max_n \left(a^nc \in \pi(f)\right)\], while $f$ is \emph{compatible with $\La_2$} if \[1+\max_n \left(a^nb \in \pi(f)\right) = \max_n \left(a^nc \in \pi(f)\right)\]
These two conditions define disjoint subsets of $H_\Sigma$, separate the two languages, and are definable through a morphism into an infinite distributive forest algebra.
Now consider the following statement

\begin{enumerate}
\item there is some nonempty path

\item each maximal path has the form $d^n a^m (b/c)$, where

\begin{enumerate}
\item the $d$'s alternatingly have children that are compatible (modulo $\pi$) with $\La_1$ and compatible with $\La_2$

\item In each path $d^n a^m b$, either (1) the $b$ node has a $c$ sibling and for each $a$ node and for the last $d$ node, what comes below it is compatible with $\La_1$, or (2) the $b$ node has an $ac$ sibling and for each $a$ node and for the last $d$ node, what comes below it is compatible with $\La_2$

\item In each path $d^n a^m c$, either (1) the $c$ node has a $b$ sibling and for each $a$ node and for the last $d$ node, what comes below it is compatible with $\La_1$, or (2) the last $a$ node exists and has a $b$ sibling and for each $a$ node before it and for the last $d$ node, what comes below it is compatible with $\La_2$
\end{enumerate}
\end{enumerate}

This can be expressed by a morphism into a wreath product of two infinite distributive algebras.
One can show that this defines $\La$:

\begin{proof}
First, it is straightforward that all forests in $\La$ satisfy this description.
For the converse, we do induction over the size of forests satisfying the description.
We want to show that, for any path where the forest below the last $d$ node is compatible with $\La_1$, those of its children whose roots are labeled $a$ actually form an element of $\La_1$. Similarly for $\La_2$.
As the base case, we take those forests that have only $d$ nodes.

Now let $f$ be a forest satisfying this description, and let $\pi = \nu_1 ... \nu_n$ a maximal sequence of nodes from root to leaf, with $\nu_1$ a root node and $\nu_n$ a leaf labeled $b$ or $c$.
We know there is at least one $d$ node, unless $f$ is an empty forest (which is covered by the base case).
Let's assume that what comes below the last $d$ node is compatible with $\La_1$.
If $\nu_n$ is labeled $b$, it has a $c$ sibling, otherwise it is labeled $c$ and has a $b$ sibling.
Thus, $\nu_n$ plus this sibling are in $\La_1$.

We remove these two nodes. We then remove any $a$ nodes that have only $a$ nodes below them.

\begin{itemize}
\item Case 1: The last $d$ node has no $a$ child. Then we can apply the induction hypothesis or go directly to the base case.

If this is not the case, what comes below the last $d$ node must be compatible with either $\La_1$ or $\La_2$.

\item Case 2: What comes below the last $d$ node is still compatible with $\La_1$. Then  we apply the induction hypothesis.

\item Case 3: What comes below the last $d$ node is now compatible with $\La_2$.
We know that there is a path $a^{l+1}c$ branching off somewhere from our path, and it is in total at most as long as our path was, and the longest $a^*b$ path is shorter by one.
Contradiction, since what comes below the node at which the two branches parted must have been compatible with $\La_1$, which contradicts what we get when considering the conditions placed on this $a^{l+1}c$ path.
\end{itemize}

We have concluded that the $a$-children of the last $d$-node must indeed have formed an element of $\La_1$.

So now let's assume that what comes below the last $d$ node is compatible with $\La_2$.
If $\nu_n$ is labeled $c$, $\nu_{n-1}$ is an $a$ node, and $\nu_{n-1}$ has a $b$ sibling.
Let's remove these two nodes, and then all $a$ nodes having only $a$ nodes below them.
Let's assume that what comes below the last $d$ node is now compatible with $\La_1$.
But then as previously we get a contradiction with conditions placed on the paths that are now longest.

In conclusion, we have shown that the description given above captures $\La$.
\end{proof}

Now let $\Ff$ be the syntactic algebra with morphism $\phi$ and let $\Gg$ be a finite distributive algebra with some morphism $\psi$.
Consider the derived category $D_{\phi,\psi}$.
Considering Proposition~\ref{prop:distr-char}, no language recognized by a finite distributive forest algebra can separate $\La_1$ from $\La_2$.
Thus, since $\Gg$ is finite and distributive, there are forests $f_1 \in \La_1$, $f_2 \in \La_2$ with $\psi(f_1) = \psi(f_2)$.
On the other hand, $d[f_1]+d[f_2] \in \La$, while $d[f_1+f_2] \not\in \La$. Thus,
$$\phi(d[f_1+f_2]) \neq \phi(d[f_1]+d[f_2])$$
Since $\psi(f_1) = \psi(f_2)$, the derived category is not locally distributive (recall Definition~\ref{def:loc-dist}).
So the derived category does not divide any distributive forest algebra.
Appealing to the second direction of the Derived Category Theorem as stated in Theorem~\ref{thm:full:deriv}, we see that, whenever $\Gg' \wr \Gg$ recognizes $\La$, the algebra $\Gg'$ cannot be distributive.

\end{document}